\documentclass[journal,onecolumn,12pt,draftclsnofoot]{IEEEtran}
\usepackage{amsmath,amsthm,amssymb}
\usepackage{algorithm,algorithmic}
\usepackage{graphicx}
\usepackage{subfigure}
\usepackage{siunitx}
\usepackage{lineno}
\usepackage{cite}
\usepackage{bm}

\def\BibTeX{{\rm B\kern-.05em{\sc i\kern-.025em b}\kern-.08em
T\kern-.1667em\lower.7ex\hbox{E}\kern-.125emX}}

\makeatletter
\newcommand{\HEADER}[1]{\ALC@it\underline{\textsc{#1}}\begin{ALC@g}}
\newcommand{\ENDHEADER}{\end{ALC@g}}
\makeatother

\newcommand{\ie}{\textit{i.e.}}

\newtheorem{thm}{Theorem}[]
\newtheorem{lem}[thm]{Lemma}
\newtheorem{rem}{Remark}
\newtheorem{prop}[thm]{Proposition}

\begin{document}
\title{Interference-Constrained Scheduling of a Cognitive Multi-hop Underwater Acoustic Network}
\author{Chen Peng,~\IEEEmembership{Student Member,~IEEE} and Urbashi Mitra,~\IEEEmembership{Fellow,~IEEE}
\thanks{This work is funded in part by one or more of the following grants: NSF CCF-1817200, ARO W911NF1910269, DOE DE-SC0021417, Swedish Research Council 2018-04359, NSF CCF-2008927, NSF CCF-2200221, ONR 503400-78050, and ONR N00014-15-1-2550.}
\thanks{The authors are with the Department of Electrical and Computer Engineering, University of Southern California, Los Angeles, CA 90089 USA (e-mail: cpeng732@usc.edu; ubli@usc.edu).}}

\maketitle

\begin{abstract}
This paper investigates optimal scheduling for a cognitive multi-hop underwater acoustic network with a primary user interference constraint. The network consists of primary and secondary users, with multi-hop transmission adopted for both user types to provide reliable communications. Critical characteristics of underwater acoustic channels, including significant propagation delay, distance-and-frequency dependent attenuation, half-duplex modem, and inter-hop interference, are taken into account in the design and analysis. In particular, time-slot allocation is found to be more effective than frequency-slot allocation due to the underwater channel model. The goal of the network scheduling problem is to maximize the end-to-end throughput of the overall system while limiting the throughput loss of primary users. Both centralized and decentralized approaches are considered. Partially Observable Markov Decision Processes (POMDP) framework is applied to formulate the optimization problem, and an optimal dynamic programming algorithm is derived. However, the optimal dynamic programming solution is computationally intractable. Key properties are shown for the objective function, enabling the design of approximate schemes with significant complexity reduction. Numerical results show that the proposed schemes significantly increase system throughput while maintaining the primary throughput loss constraint. Under certain traffic conditions, the throughput gain over frequency-slot allocation schemes can be as high as $50\%$.
\end{abstract}

\begin{IEEEkeywords}
Cognitive underwater acoustic network, multi-hop transmission, POMDP, dynamic programming, approximation scheme
\end{IEEEkeywords}

\section{Introduction}
\IEEEPARstart{U}{nderwater} communications and networks are currently an active research area due to a wide range of applications, including environment monitoring, oceanography data collection, underwater target tracking, \textit{etc.} \cite{luo2014challenges}. As electromagnetic and optical signals have limited transmission ranges, acoustic signals are employed for underwater wireless communication \cite{partan2007survey}. In contrast to terrestrial communications, underwater acoustic communications (UAC) possesses distinguishing characteristics, such as significant propagation delay, frequency-dependent attenuation, and limited spectrum bandwidth \cite{stojanovic2009underwater}.

Various underwater acoustic systems, including natural and artificial, use acoustic signals for communication, sensing, detection, and localization. Due to frequency-dependent attenuation, most natural and artificial systems utilize a narrow bandwidth between $1$ kHz and $100$ kHz \cite{luo2014challenges}, leading to a potentially overcrowded spectrum. Nevertheless, the spectrum can still be underutilized both temporally and spatially. Cognitive radio approaches have been proposed in underwater acoustic networks (UANs) to achieve efficient spectrum utilization \cite{luo2014challenges,mishachandar2021underwater,baldo2008cognitive,pottier2017robust,demirors2015software,geng2022exploiting}. Spectrum sensing, dynamic power control, and spectrum management are typical components of cognitive acoustic systems. Additionally, many UAC systems utilize multi-hop transmission to provide broader area coverage with high efficiency, where a long distance is divided into multiple hops \cite{stojanovic2007relationship,carbonelli2009error,zhang2010analysis}. As the ongoing and future applications of UAN are likely to require larger area coverage and the coexistence of heterogeneous systems, there is a need to develop cognitive scheduling schemes for multi-hop UANs.

Key aspects of our proposed framework have been previously considered. Spectrum allocation and utilization for cognitive UANs has been investigated for the single hop case \cite{mishachandar2021underwater,baldo2008cognitive,pottier2017robust}. In \cite{baldo2008cognitive}, a frequency channel allocation scheme is proposed, which exploits user location knowledge to maximize the minimum channel capacity. In \cite{pottier2017robust}, the authors propose a decentralized spectrum-sharing method for orthogonal frequency-division multiplexing systems in interference channels, where the problem is formulated as a non-cooperative game. While spectrum allocation is commonly considered in radio-frequency-based cognitive radio systems, it may not be an efficient strategy for UAN due to the limited bandwidth and severe signal attenuation as a function of range and frequency. On the other hand, the large propagation delay provides a unique interference management opportunity in the time domain. Thus, herein, we explore underutilized resources in the time domain to develop our scheduling schemes. The numerical results verify that cognitive time-slot allocation can be more efficient than cognitive frequency-slot allocation in UANs.

Packet size is another critical parameter of UAC (see {\em e.g.} \cite{stojanovic2005optimization}). Optimization of packet size has been studied for UAC based on various performance metrics. In \cite{basagni2012optimized}, the impact of packet size on the performance of media access control (MAC) protocols for UAN is investigated. An algorithm is proposed in \cite{ayaz2012reliable} to determine the suitable data packet size for increasing the transmission reliability in UAN. For cognitive radio networks, a dynamic packet size optimization and channel selection scheme is proposed in~\cite{jamal2015dynamic}, where the framework of the constrained Markov decision process is employed for problem solution. However, these prior works do not actively exploit a key unique feature of UAC, the significant propagation delay. Herein, we show that these delays can be effectively leveraged to improve throughput in UANs.

Multi-hop cooperative transmission and network scheduling for UAC systems has attracted intense research interest \cite{chen2007delay,carbonelli2009error,zhang2010analysis,basagni2010choosing,cao2010capacity,lmai2015throughput,de2016optimizing}. In \cite{zhang2010analysis}, performance bounds and scheduling design for multi-hop underwater acoustic line networks were devised. The transmission schedule that maximizes network throughput for underwater acoustic multi-hop grid networks is investigated in \cite{lmai2015throughput}. Optimal parameter selection, including the number of hops, re-transmission, and code rate, is analyzed in \cite{de2016optimizing} to improve energy efficiency for multi-hop UAC. However, the scheduling problem is more complicated than direct parameter optimization and requires an alternative formulation in a cognitive multi-hop UAN.

Recently, two works have investigated the time domain UAN scheduling problem. In \cite{zeng2017distributed}, a TDMA-based scheduling scheme for UAN is developed to achieve free of interference communication. The interference alignment is achieved by greedily maximizing interference overlapping possibilities. Although interference-free communication is desirable, it also greatly restricts the potential of improving network performance, especially when there is high uncertainty. In \cite{geng2022exploiting}, a deep-reinforcement learning-based MAC is proposed to improve the throughput of UAN, where one agent node coexists with other nodes employing traditional protocols. These works focus on improving the performance of UAN where users have the same priority level. The coexistence of primary users (PUs) and secondary users (SUs) necessitates a new mechanism to ensure the performance of PUs. Specifically, we consider the decentralized interference management problem for users with priority differences in UANs, under the POMDP framework.

Herein, we introduce cognitive radio ideas into the scheduling problem for multi-hop UANs. This joint problem formulation offers challenges with an increased decision dimension and the need for interference management. However, the multi-hop feature offers additional temporal/spatial reuse opportunities. We consider a multi-hop network of PUs that coexists with a multi-hop network of SUs. The main contributions of this work are summarized as follows:
\begin{itemize}
\item We model the time-evolving behavior of cognitive multi-hop networks in the realistic underwater acoustic channel as a discrete-time Markov chain, which captures the distinguishing underwater channel characteristics of long propagation delay and high path attenuation.
\item We formulate the centralized scheduling problem with a primary user interference constraint as a constrained Partially Observable Markov Decision Process (POMDP) with the existence of a central controller. A dynamic programming (DP) algorithm is derived for solving the optimization problem.
\item We prove key properties of the reward-to-go functions, which enable the derivation of an efficient approximation scheme for the centralized system. The performance of the centralized scheme provides a benchmark for the decentralized scheme to be developed.
\item We formulate the decentralized scheduling problem with a primary user interference constraint as a constrained, decentralized Partially Observable Markov Decision Process (Dec-POMDP). An efficient approximation algorithm is derived for solving the problem.
\item We investigate the impact of SU packet sizes on interference levels. By utilizing the unique relationship between interference and packet sizes, an optimization problem of SU packet sizes is formulated and solved.
\item We evaluate the performance of the proposed approximation schemes through simulation results and illustrate that they are adaptive to the unique characteristics of underwater acoustic channels.
\end{itemize}

The rest of this paper is organized as follows. In Section \ref{SM}, the system model is described. The centralized scheduling problem with performance constraint is formulated in Section \ref{OCSP} and solved by a proposed approximation scheme in Section \ref{CCTS}. In Section \ref{ODSP}, we formulate the decentralized scheduling problem and propose an approximation scheme in Section \ref{DCTS}. Performance of both centralized and decentralized schemes are evaluated via numerical results in Section \ref{NR}. Finally, we conclude the paper in Section \ref{CC}.

\section{System Model} \label{SM}
In this section, we describe the underwater channel model, state space representation of the system and transition probabilities for the cognitive multi-hop network.
\subsection{Underwater Acoustic Channel Model}
We adopt the underwater acoustic channel model introduced in \cite{qarabaqi2013statistical}, which incorporates both small-scale variations (due to scattering) and large-scale channel variations (due to multipath and location uncertainty). The attenuation $A(d,f)$ experienced by an underwater acoustic signal transmitted
over a distance $d$ meters at carrier frequency $f$ (in \unit{kHz}), is given by (in \unit{dB}):
\begin{equation}
10\log_{10}(A(d,f)/A_0) = \kappa 10\log_{10}d + \frac{d}{10^3} 10\log_{10} a(f),
\label{eq:attenuation}
\end{equation}
where $A_0$ is a normalizing constant, $\kappa$ denotes the spreading factor and $a(f)$ is the absorption coefficient \cite{stojanovic2007relationship}. The absorption coefficient can be expressed in \unit{dB/km} using Thorp's empirical formulas for $f$ in kilohertz as \cite[Chapter~1]{brekhovskikh2003fundamentals}
\begin{equation}
10\log_{10}a(f) = 0.11\frac{f^2}{1+f^2} + 44\frac{f^2}{4100+f^2} + 2.75 \times 10^{-4} f^2 + 0.003.
\end{equation}

Considering multiple paths of nominal length $\bar{d}_l, l = 0,\dots,L-1$ with corresponding transfer functions modeled as
\begin{equation}
\bar{H}_l(f) = \frac{\Gamma_l}{\sqrt{A(\bar{d}_l, f)}} = \frac{\Gamma_l}{\sqrt{\left(\bar{d}_l/\bar{d}_0\right)^\kappa a(f)^{\bar{d}_l-\bar{d}_0}}} \bar{H}_0(f),
\end{equation}
where $\Gamma_l$ is the cumulative reflection coefficient and the relative nominal path gain is given by $\bar{h}_l = \Gamma_l/\sqrt{(\bar{d}_l/\bar{d}_0)^\kappa a(f)^{\bar{d}_l-\bar{d}_0}}$. Large-scale displacements cause the path length to deviate from the nominal as $d_l = \bar{d}_l + \Delta d_l$, where $\Delta d_l$ is considered to be random. Based on the approximation introduced in \cite{qarabaqi2013statistical}, the relative path gain can be expressed as
\begin{equation}
h_l \approx \bar{h}_l e^{-\xi_l \Delta d_l/2},
\end{equation}
where $\xi_l = a_0 - 1 + \kappa/\bar{d}_l$ and $a_0$ is the absorption factor corresponding to a frequency within the signal bandwidth. Further, considering the small-scale effect of scattering, a path is split into a number of micro-paths. The overall channel transfer function is expressed as
\begin{equation}
H(f) = \bar{H}_0(f) \sum_l \sum_i h_{l,i} e^{-j2\pi f \tau_{l,i}} = \bar{H}_0(f) \sum_l h_l \gamma_l(f) e^{-j2\pi f \tau_p},
\end{equation}
where $\tau_l$ are the path delays and $\tau_{l,i} = \tau_l + \delta\tau_{l,i}$ are the micro-path delays. The small-scale fading coefficient is defined as
\begin{equation}
\gamma_l(f) \triangleq \frac{1}{h_l}\sum_i h_{l,i} e^{-j2\pi f \delta\tau_{l,i}},
\end{equation}
where the delays $\delta\tau_{l,i}$ are considered as random variables. 
We define the channel gain, $G,$ for a system operating in frequency range $[f_c - B/2, f_c + B/2]$ as
\begin{equation}
G = \frac{1}{B} \int_{f_c - B/2}^{f_c + B/2} |H(f)|^2 df,
\end{equation}
which can be modeled as a log-normal random variable (see \cite{qarabaqi2013statistical} and references therein). The mean and variance of the approximate log-normal channel gain can be calculated theoretically \cite{cardieri2000statistics} or estimated from experimental data.

The ambient noise in the ocean consists of four sources: turbulence, shipping, waves and thermal noise, which is modeled as Gaussian \cite{stojanovic2007relationship}. The power spectral densities of the four noise components in \unit{dB} re \unit{\micro\pascal/\hertz} (\ie, the power per unit bandwidth associated with the reference sound pressure level of 1$\mu$Pa) can be formulated as functions of frequency in kHz \cite{stojanovic2007relationship}
\begin{align}
10 \log_{10} N_\text{t}(f) &= 17 -30\log_{10} f, \nonumber \\
10 \log_{10} N_\text{s}(f) &= 40 + 20(s-0.5) + 26\log_{10} f - 60\log_{10}(f+0.03), \nonumber \\
10 \log_{10} N_\text{w}(f) &= 50 + 7.5\sqrt{w} + 20\log_{10} f - 40 \log_{10}(f+0.4), \nonumber \\
10 \log_{10} N_\text{th}(f) &= -15 + 20\log_{10} f,
\end{align}
where $s$ is the shipping activity factor and $w$ is the wind speed in \unit{m/s}. The overall power spectral density of noise is
\begin{equation}
N(f) = N_\text{t}(f) + N_\text{s}(f) + N_\text{w}(f) + N_\text{th}(f).
\end{equation}

\subsection{State Space Representation}
We consider a cognitive multi-hop (CM) UAN consisting of multi-hop chains of $N_P$ PU hops and $N_S$ SU hops respectively, as shown in Fig. \ref{fig:network}. In the multi-hop network consisting of both PUs and SUs, the transmitter at one end of the multi-hop chain sends packets to the receiver at the other end, with the help of relay nodes in between. In the system, PUs adopt a fixed, decode-and-forward protocol. Due to the half-duplex nature of acoustic modems, and to avoid the hidden terminal problem \cite{al2020survey}, each PU will stay idle for two time slots after the transmission of a packet. We define the system state as the current action performed by each PU and the current buffer state of each SU. The channel quality information is not included in the system state, but reflected in the state transition probabilities introduced in the sequel. The underlying state vector in time slot $t$ for the PUs is represented as
\begin{equation}
\bm{s}_t \triangleq (s_t^1, s_t^2, \dots, s_t^{N_P})^T \in \bm{S} \equiv \{0,1\}^{N_P},
\end{equation}
\noindent where $s_t^j$ is the state of the $j$-th PU hop, which can be idle ($s_t^j = 0$) or transmitting ($s_t^j = 1$). The buffer state vector of the SUs in time slot $t$ is represented as
\begin{equation}
\bm{b}_t \triangleq (b_t^1, b_t^2, \dots, b_t^{N_S})^T \in \bm{B} \equiv \{0,1\}^{N_S},
\end{equation}
where $b_t^i$ is the state of the $i$-th SU hop, which can be empty ($b_t^i = 0$) or nonempty ($b_t^i = 1$). 

\begin{figure}[!t]
\centering
\includegraphics[width=0.7\linewidth]{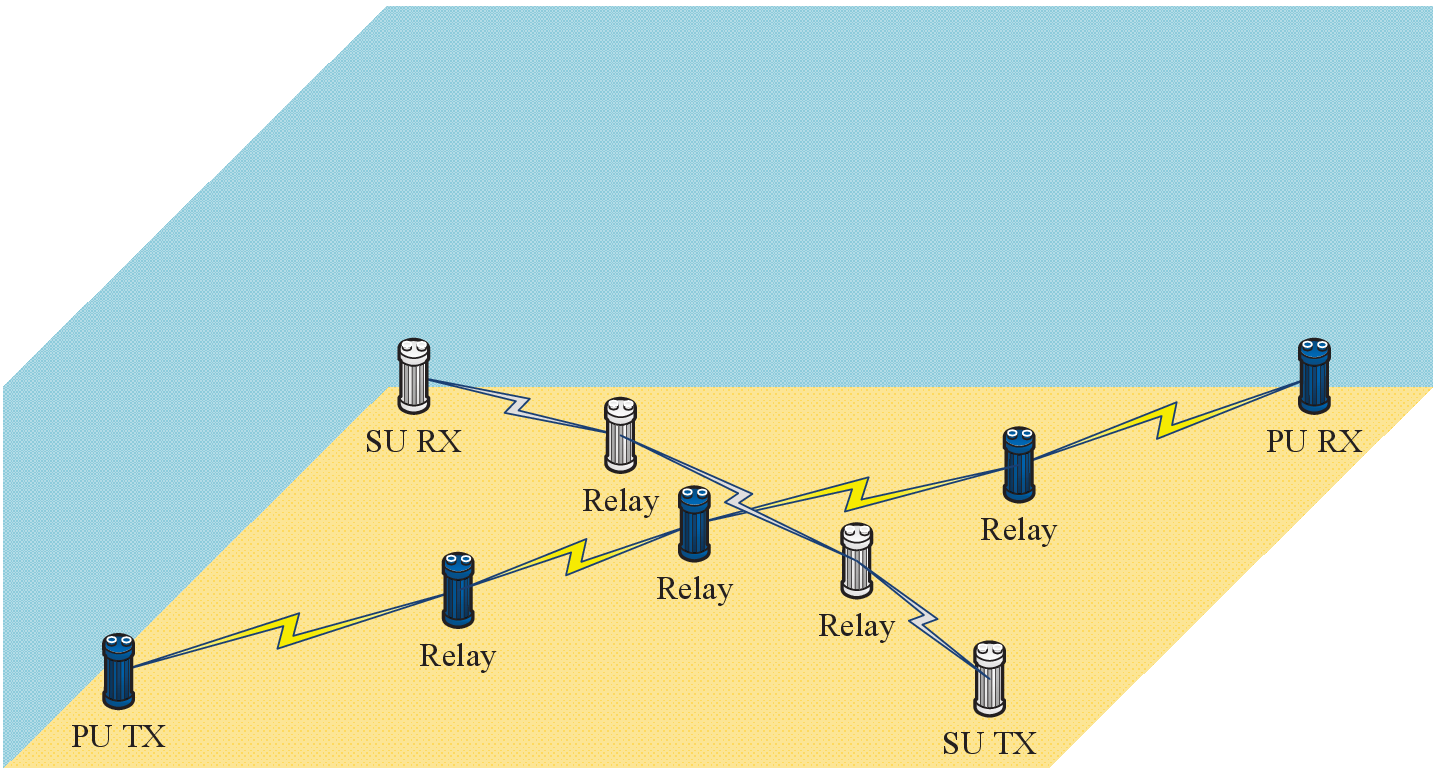}
\caption{A cognitive multi-hop underwater acoustic network of $N_P = 4$ PU hops and $N_S = 3$ SU hops.}
\label{fig:network}
\end{figure}

During in total $T$ time slots, the SUs opportunistically access the spectrum to transmit packets. The decision vector of SUs in each time slot is represented as
\begin{equation}
\bm{\delta}_t \triangleq (\delta_t^1, \delta_t^2, \dots, \delta_t^{N_S})^T \in \bm{\Delta} \equiv \{0,1\}^{N_S},
\end{equation}
where $\delta_t^i$ is the decision variable for the $i$-th SU, which can be either idle ($\delta_t^i = 0$) or transmitting ($\delta_t^i = 1$). The buffer state vector contains information about the existence of packets in the SU buffers while the decision vector indicates whether the SUs decide to transmit packets or not. Note that $\bm{\delta}_t$ is conditioned on $\bm{b}_t$, because a SU can only transmit when its buffer is not empty. Different from terrestrial cognitive radio networks, SUs can not conduct both sensing and transmission effectively in one slot, because of the large propagation delay characteristic of underwater acoustic channel.

\subsection{Transition Probability} \label{SM-TP}
The system state transition depends on two factors, new packet arrival and packet transmission loss. We model the packet traffic of PUs as a two-state Markov process, depicted in Fig. \ref{fig:arrival_MP}, that alternates between on-period (packet arriving) and off-period (idle) \cite{zhou2008discrete}. The traffic density is thus characterized by two parameters: $\alpha_1$, $\alpha_2 \in [0,1]$. Taking the bursty packet arrival nature into account, we assume $\alpha_1 < \alpha_2$, meaning that packets tend to arrive in succession. The SU transmitter is assumed to be backlogged, \textit{i.e.}, its buffer is always full.

\begin{figure}[!t]
\centering
\includegraphics[width=0.5\linewidth]{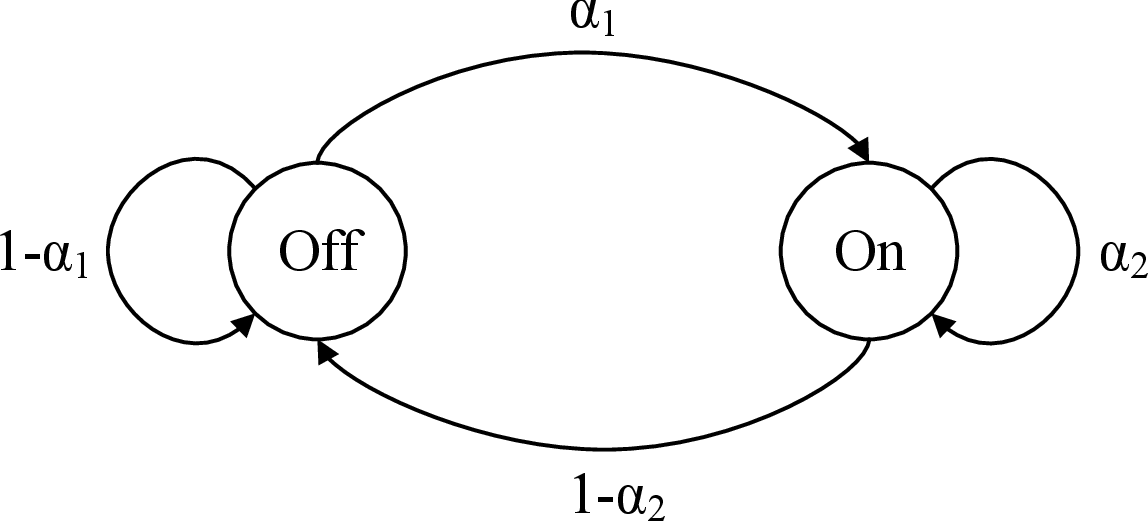}
\caption{The two-state Markov model for packet arrival.}
\label{fig:arrival_MP}
\end{figure}

Another factor that determines state transition probability is the packet loss rate. For each relay node of the PUs, if a packet is received without error, it transmits the packet to the next node with probability $1$. The packet loss probability is determined by per bit signal to interference and noise ratio (SINR), which is a function of random channel gain, interference strength, and noise. We formulate the per bit SINR as
\begin{equation}
\gamma_{P,k}^j(\bm{s}_t, \bm{\delta}_t) = \frac{G \cdot P_{TX}}{I_{P,k}^j(\bm{s}_t,\bm{\delta}_t) + \int N(f) df},
\end{equation}
where $P_{TX}$ is the transmit power and $I_{P,k}^j(\bm{s}_t, \bm{\delta}_t)$ is the interference power from other PUs and SUs, on the $k$-th bit. As aforementioned, the channel gain $G$ follows a log-normal distribution, which enables us to approximate the conditional SINR as log-normal distribution with mean $\mu_{\ln\gamma}$ and variance $\sigma^2_{\ln\gamma}$ \cite[Chapter~3]{stuber2001principles}. Assuming QPSK transmission, the average error rate of the $k$-th bit can be computed by numerical integration as
\begin{equation}
P_b^k(\bm{s}_t,\bm{\delta}_t) = \int_{0}^{\infty} \frac{1}{\gamma \sigma_{\ln\gamma}(k)\sqrt{2\pi}}\exp\left[-\frac{(\ln\gamma-\mu_{\ln\gamma}(k))^2}{2\sigma^2_{\ln\gamma}(k)}\right] \cdot Q(\sqrt{2\gamma}) d\gamma,
\label{eq:ber}
\end{equation}
where $Q(\cdot)$ is the standard normal tail distribution function. For other modulation schemes, the average bit error rate can be calculated similarly. Notice that $P_b^k$ is a function of state $\bm{s}_t$ and decision $\bm{\delta}_t$ because $\mu_{\ln\gamma}(k)$ and $\sigma^2_{\ln\gamma}(k)$ are functions of them. Given average bit error rate, we can approximate the packet loss probability of the $j$-th PU hop in time slot $t$ by
\begin{equation}
p^j_P(\bm{s}_t, \bm{\delta}_t) \approx 1 - \prod_{k=1}^{L_P}\left[1 - P_b^k(\bm{s}_t,\bm{\delta}_t)\right],
\end{equation}
where $L_P$ is the PU packet size in bits. The packet loss probability for SUs, denoted as $p_S^i$, can be computed in a similar way.

Based on the packet arrival model and the packet loss model, we model the temporal evolution of the PU state $\bm{s}_t$ as a Markov chain, with transition matrix $\bm{P}(\bm{\delta}_t)$, which is a function of SU decision vector $\bm{\delta}_t$ in each time slot.

Notice that \eqref{eq:ber} provides the bit error rate (BER) of the $k$-th bit; different bits in a packet may have very different BERs. These differences arise from a combination of UAN channel effects: fading, multipath and the long propagation delay. Due to propagation delay differences between the desired signal and the interference signal, only part of the desired signal is interfered. Therefore it is possible that the overlapped part is small enough (or no overlap occurs) that the packet can be received with limited error. This opportunity is captured in the transition probabilities and will be leveraged to improve throughput in UANs in this work.

\section{The Optimal Centralized Scheduling Problem} \label{OCSP}
In this section, we formulate the interference-constrained scheduling problem for a centralized SU network. Although centralized control may not be practical, the proposed centralized scheme provides a benchmark for the decentralized scheme to be developed in the sequel, which is a key focus of the current work. Determining the location of the central controller and scheduling control signals are design problems. Herein, we assume that there is a central controller (CC) such that the overhead of information exchange between SUs and the CC is one time slot. Also, we assume control signals for gathering observation and broadcasting decisions are scheduled in dedicated bands. Under these optimistic assumptions, we derive an upper bound on the performance of the centralized scheme.

\subsection{Observation Model and Sufficient Statistics}
Scheduling of the SUs is based on channel occupancy measurements collected by the SUs in the sensing mode. We assume the channel sensing is conducted in a centralized fashion, such that each SU collects noisy measurements independently and reports to the central controller. The observation model for the $i$-th SU in time slot $t$ is expressed as
\begin{equation}
y_t^i = \bm{c}_i^T \bm{s}_t + \bm{d}_i^T \bm{\delta}_t + n_t^i,
\label{eq:obs_model}
\end{equation}
where $n_t^i \sim \mathcal{N}(0,\sigma_N^2)$ is the i.i.d. Gaussian noise over time and across SUs. $\bm{c}_i$ and $\bm{d}_i$ are measurement vectors for PUs and other SUs that are within the one slot sensing range, respectively, containing channel attenuation information.

The information vector in time slot $t$, denoted as $\bm{I}_t$, can be expressed as
\begin{equation}
\bm{I}_t = (\bm{y}_1,\dots,\bm{y}_{t-1},\bm{\delta}_1, \dots,\bm{\delta}_{t-1}),
\end{equation}
which is of expanding dimension. To reduce the information that are indeed necessary for decision making, we use the conditional probability distribution of state $\bm{s}_t$, denoted by $\bm{\omega}_t$, which is a sufficient statistic for $\bm{I}_t$ \cite[Chapter~5]{bertsekas2012dynamic}. $\bm{\omega}_t$ is also known as belief state and is defined as
\begin{equation}
\bm{\omega}_t = [\dots, \omega_t(\bm{s}), \dots]^T, \text{ where } \omega_t(\bm{s}) \triangleq \mathbb{P}(\bm{s}_t = \bm{s}|\bm{I}_t), \bm{s} \in \bm{S}.
\end{equation} 

Given observation $\bm{y}_t$, the central controller can update its belief state via the Bayes' Rule as follows
\begin{equation}
\omega_t(\bm{s}) \triangleq \mathbb{P}(\bm{s}_t = \bm{s}|\bm{I}_t) = \frac{\left[\prod_{i=1}^{N_S}\mathbb{P}(y_t^i|\bm{s}_t=\bm{s})\right]\left[\bm{\omega}_{t-1}^T \bm{P}(\bm{\delta}_{t-1})_{\bm{s}}\right]}{\sum_{\bm{s}' \in \bm{S}}\left[\prod_{i=1}^{N_S}\mathbb{P}(y_t^i|\bm{s}_t=\bm{s}')\right]\left[\bm{\omega}_{t-1}^T \bm{P}(\bm{\delta}_{t-1})_{\bm{s}'}\right]},
\label{eq:belief_update_C}
\end{equation}
where $\bm{P}(\bm{\delta}_{t-1})_{\bm{s}}$ represents the column of the transition matrix $\bm{P}(\bm{\delta}_{t-1})$ corresponding to state $\bm{s}$ and $\mathbb{P}(y_t^i|\bm{s}_t=\bm{s})$ is the probability distribution function of a Gaussian distribution evaluated at $y_t^i$. Further, we denote the belief state update rule as $\Phi(\cdot)$, then $\bm{\omega}_t = \Phi(\bm{\omega}_{t-1},\bm{\delta}_{t-1},\bm{y}_t)$.

Notice that because of the large propagation delay, the observation takes one time slot, and the information exchange between SUs and the central controller takes another time slot. Therefore the central controller has to make a decision for time slot $t$ only based on the belief $\bm{\omega}_{t-2}$ calculated from the previous observation. Due to imperfect observations (noisy and outdated), the central controller has to make decisions based on the partially observed state information. Thus, the optimal scheduling of a CM-UAN is a POMDP problem \cite[Chapter~6]{kochenderfer2015decision}. The uncertainty existing in the time-evolving system consists of the following sources: random packet arrival, random transmission failure and imperfect observation with delay.

\subsection{Problem Formulation} 
Our goal is to maximize the expected end-to-end throughput of the overall system over a finite horizon of $T$ time slots. We denote the one slot throughput of PU and SU networks as $g_P(\cdot)$ and $g_S(\cdot)$, respectively. The total throughput in slot $t$ is defined as follows
\begin{align}
g(\bm{s}_t, \bm{\delta}_t) &= g_P(\bm{s}_t, \bm{\delta}_t) + g_S(\bm{s}_t, \bm{\delta}_t) \\
&\stackrel{(a)}{=} L_P \cdot s_t^{N_P} p_P^{N_P}(\bm{s}_t,\bm{\delta}_t) + L_S \cdot \delta_t^{N_S} p_S^{N_S}(\bm{s}_t,\bm{\delta}_t),
\end{align}
where $L_S$ is the SU packet size in bits. Equality (a) holds because the throughput is achieved only when the transmission to the end receiver is successful. Denote a centralized policy as $\Psi_c$, which is a mapping from the belief state to a decision,
\begin{equation}
\bm{\delta}_t = \Psi_c(\bm{\omega}_t).
\end{equation}
Furthermore, denote the expected accumulated throughput of both PUs and SUs over a finite horizon as $V$, which can be decomposed as $V_P + V_S$. The optimization problem can be formulated as follows
\begin{align}
&\text{Maximize:} & &V(\Psi_c) = \mathbb{E}\left\{\sum_{t=1}^T g(\bm{s}_t, \bm{\delta}_t) \bigg|\Psi_c \right\} \label{obj} \\
&\text{Subject to:} & &V_P(\Psi_c) \geq \beta\cdot V_P(\Psi_0),
\label{constr}
\end{align}
where $\beta \in (0,1)$ is the PU throughput degradation coefficient. $\Psi_0$ is the policy of always being silent and the definition of $V_P(\Psi_0)$ is given as
\begin{equation}
V_P(\Psi_0) \triangleq \mathbb{E}\left\{\sum_{t=1}^T g_P(\bm{s}_t, \bm{\delta}_t) \bigg|\Psi_0 \right\} = \mathbb{E}\left\{\sum_{t=1}^T g_P(\bm{s}_t, \bm{0}) \right\}.
\end{equation}
The constraint thus states that the expected throughput of PUs, given by the optimal decision sequence, should be no less than the $\beta$ weighted throughput of PUs in the case of no SU interference. In this way, the interference level of SUs can be controlled by selecting specific value of $\beta$.

The optimization problem is a constrained POMDP problem and we make use of the belief state $\bm{\omega}_t$ to rewrite the objective function as
\begin{equation}
V(\Psi_c) = \mathbb{E}\left\{\sum_{t=1}^T \bm{\omega}_t^T \cdot \bm{g}(\bm{\delta}_t) \bigg| \Psi_c \right\},
\end{equation}
where $\bm{g}(\bm{\delta}_t)$ is the vector of total throughput at each state, defined as
\begin{equation}
\bm{g}(\bm{\delta}_t) \triangleq L_P \cdot s_t^{N_P} \bm{p}_P^{N_P}(\bm{\delta}_t) + L_S \cdot \delta_t^{N_S} \bm{p}_S^{N_S}(\bm{\delta}_t),
\label{eq:th_C}
\end{equation}
where $\bm{p}_P^{N_P}(\bm{\delta}_t)$ and $\bm{p}_S^{N_S}(\bm{\delta}_t)$ are the vectors of packet success rates, consisting respectively of $p_P^{N_P}(\bm{s}_t,\bm{\delta}_t)$ and $p_S^{N_S}(\bm{s}_t,\bm{\delta}_t)$, $\forall \bm{s}_t \in \bm{S}$.

This constrained optimization problem can be solved by dynamic programming \cite[Chapter~2]{bertsekas2012dynamic}. Taking the constraint into account, we formulate the reward-to-go function $V_t$ as a function of both the belief state $\bm{\omega}_t$ and the achieved PU throughput $\eta_t$. For $t = T-1, \dots, 0$, the reward-to-go function $V_t(\bm{\omega}_t, \eta_t)$ is related to $V_{t+1}(\bm{\omega}_{t+1},\eta_{t+1})$ through the recursion
\begin{equation}
V_t(\bm{\omega}_t,\eta_t) = \max_{\bm{\delta}_t \in \bm{\Delta}} \biggl\{\bm{\omega}_t^T \cdot \bm{g}(\bm{\delta}_t) + \mathop{\mathbb{E}}_{\bm{y}_{t+1}}\left[V_{t+1}\left(\Phi(\bm{\omega}_t, \bm{\delta}_t, \bm{y}_{t+1}), \eta_t - \bm{\omega}_t^T \bm{g}_P(\bm{\delta}_t)\right)\right] \biggr\},
\end{equation}
\noindent where $\Phi(\cdot)$ is the belief state update function. The reward-to-go function for $t = T$ is given by
\begin{equation}
V_T(\bm{\omega}_T, \eta_T) = \max_{\bm{\delta}_T \in \bm{\Delta}}\left\{\bm{\omega}_T^T \cdot \bm{g}(\bm{\delta}_T) \right\}.
\end{equation}
However, solving the dynamic programming problem is intractable because of the uncountably infinite belief space, an exponentially large control space and the uncountably infinite observation space. Specifically, for $m$ possible system states, $d$ decision choices, and a quantization of the belief probability with $n_1$ levels, the number of belief states is $\mathcal{O}(n_1^m)$, where $m = \mathcal{O}(2^{N_P + N_S})$. Further, given a quantization of the observation space with $n_2$ levels, and the constraint value with $n_3$ levels, the computational complexity for determining the optimal decision sequence is $\mathcal{O}(d \cdot n_1^m n_2 n_3 \cdot T)$.

\section{Centralized Cognitive Time Slot Scheduling} \label{CCTS}
In the previous section, we argued that optimal scheduling using dynamic programming is computationally intractable. In this section, we present a low-cost approximation scheme that can be implemented and run in real-time. We denote this method as the \textit{Centralized Cognitive Time Slot Scheduling (CCTS)} scheme. We begin by reviewing the factors that contribute to the computational complexity: an average performance constraint and imperfect observation with delay. By exploring the unique characteristics of the system model and properties of the reward-to-go function, we develop the approximation scheme.

\subsection{Apply Local instead of Global Constraints}
We take the first step to revise the constraint over the entire horizon, which greatly complicates the dynamic programming solution due to an expanding dimension of the solution space. To this end, we approximate the original problem by posing stricter local constraints such that in each time slot $t$, the PU throughput loss level calculated from slot $t$ to the last slot $T$ must always be lower than the degradation coefficient, which is expressed as
\begin{equation}
V_{P,t} = \mathbb{E}\left\{\sum_{\tau=t}^T g_P(\bm{s}_\tau, \bm{\delta}_\tau) \bigg|\Psi_c \right\} \geq \beta \cdot V_{P,t}(\Psi_0), \quad \forall t \in \{1,2,\dots,T\},
\label{approx_constr}
\end{equation}
where $V_{P,t}(\Psi_0) = \mathbb{E}\left\{\sum_{\tau=t}^T g_P(\bm{s}_\tau, \bm{0}) \right\}$. The original constraint \eqref{constr} only requires that the inequality holds for $t=1$ while the stricter constraint \eqref{approx_constr} requires it holds $\forall t \in \{1,2,\dots,T\}$.  

\begin{prop}
The feasible set of \eqref{approx_constr} is not empty and whenever a sequence of decisions satisfies \eqref{approx_constr}, it also satisfies~\eqref{constr}.
\end{prop}

\begin{proof}
We prove this proposition by induction. First, for $t = T$, there exists feasible decisions that satisfy \eqref{approx_constr} because $\bm{\delta}_T = \bm{0}$ is always a valid choice. Next, let us assume that for $t = k+1$, a valid sequence of decisions, $\bm{\delta}_\tau$ ($k+1\leq \tau \leq T$), satisfying $\mathbb{E}\left\{\sum_{\tau=k+1}^T g_P(\bm{s}_\tau, \bm{\delta}_\tau) \big|\Psi_c \right\} \geq \beta \cdot \mathbb{E}\left\{\sum_{\tau=k+1}^T g_P(\bm{s}_\tau, \bm{0}) \right\}$ is found. Then for $t = k$, we want to select $\bm{\delta}_k$ such that the following expression is positive,
\begin{align}
&\mathbb{E}\left\{\sum_{\tau=k}^T g_P(\bm{s}_\tau, \bm{\delta}_\tau) \bigg|\Psi_c \right\} - \beta \cdot \mathbb{E}\left\{\sum_{\tau=k}^T g_P(\bm{s}_\tau, \bm{0}) \right\} \quad \\
&= \mathbb{E}\left\{g_P(\bm{s}_k, \bm{\delta}_k) + \sum_{\tau=k+1}^T g_P(\bm{s}_\tau, \bm{\delta}_\tau^*) \bigg|\Psi_c \right\} - \beta \cdot \mathbb{E}\left\{g_P(\bm{s}_k, \bm{0}) + \sum_{\tau=k+1}^T g_P(\bm{s}_\tau, \bm{0}) \right\} \\
&\geq \mathbb{E}\left\{g_P(\bm{s}_k, \bm{\delta}_k) \big|\Psi_c \right\} - \beta \cdot \mathbb{E}\left\{g_P(\bm{s}_k, \bm{0}) \right\}.
\label{constr_m}
\end{align}
Such $\bm{\delta}_k$ always exists because at least $\bm{\delta}_k = \bm{0}$ is a feasible decision. Thus by induction, a sequence of decisions satisfying \eqref{approx_constr} can be found for $\forall t \in \{1,2,\dots,T\}$. When $t = 1$, \eqref{approx_constr} is the same as the original constraint, which is satisfied by any sequence of decisions that satisfies \eqref{approx_constr}.
\end{proof}

With stricter constraints, the approximated optimization problem must have lower or equal objective value. Thus this approximation provides a lower bound for $V_t(\bm{\omega}_t,\eta_t)$, denoted as $\underline{V}_t(\bm{\omega}_t)$, with the relationship stated as follows
\begin{equation}
\underline{V}_t(\bm{\omega}_t) \leq V_t(\bm{\omega}_t,\eta_t), \hspace{8pt} \forall \eta_t \leq \mathbb{E}\left\{\sum_{\tau=t}^T g_P(\bm{s}_\tau, \bm{\delta}_\tau) \bigg|\Psi_c \right\}.
\end{equation}
This lower bound has a simpler expression in terms of a recursion:
\begin{equation}
\underline{V}_t(\bm{\omega}_t) = \max_{\bm{\delta}_t \in \bm{\Delta}_f}\Bigl\{\underbrace{\bm{\omega}_t^T \bm{g}(\bm{\delta}_t)}_{\textstyle\text{immediate}} + \underbrace{\mathop{\mathbb{E}}\left[V_{t+1}(\Phi(\bm{\omega}_t,\bm{\delta}_t, \bm{y}_{t+1}))\right]}_{\textstyle\text{future}}\Bigr\},
\label{approx_recur}
\end{equation}
where $\bm{\Delta}_f$ is the set of feasible decisions satisfying \eqref{approx_constr}.

\subsection{An Upper Bound of the Reward Function}
After reducing the complexity caused by the global constraint, the optimization problem still remains computationally expensive to solve, because of the uncountably infinite belief space. To resolve this issue, we explore the convexity of the reward-to-go function.

\begin{lem}
The function $\underline{V}_t(\bm{\omega}_t)$ is a convex function.
\end{lem}

\begin{proof}
For $t = T$, the throughput given by each possible decision choice, $\bm{\omega}_T^T \cdot \bm{g}(\bm{\delta}_T)$ is a linear function of $\bm{\omega}_T$, and the point-wise maximum of linear functions, $\underline{V}_T(\bm{\omega}_T)$, is a convex function. Next, assume $\underline{V}_k(\bm{\omega}_k)$ ($2 \leq k \leq T-1$) is a convex function. The form of $\underline{V}_{k-1}(\bm{\omega}_{k-1})$ is as follows
\begin{align}
&\underline{V}_{k-1}(\bm{\omega}_{k-1}) = \max_{\bm{\delta}_{k-1}} \left\{\bm{\omega}_{k-1}^T \bm{g}(\bm{\delta}_{k-1}) + \mathop{\mathbb{E}}\left[\underline{V}_k(\bm{\omega}_k)\right] \right\} \nonumber \\
&= \max_{\bm{\delta}_{k-1}} \left\{\bm{\omega}_{k-1}^T \bm{g}(\bm{\delta}_{k-1}) + \bm{\omega}_{k-1}^T \mathop{\mathbb{E}}\left[\underline{V}_k(\bm{\omega}_k)|\bm{s}_{k-1}=\bm{s} \right] \right\} \nonumber \\
&= \max_{\bm{\delta}_{k-1}} \left\{\bm{\omega}_{k-1}^T \cdot \bigl[ \bm{g}(\bm{\delta}_{k-1}) + \mathop{\mathbb{E}}\left[\underline{V}_k(\bm{\omega}_k)|\bm{s}_{k-1}=\bm{s} \right] \bigl] \right\},
\end{align}
where $\mathop{\mathbb{E}}\left[\underline{V}_k(\bm{\omega}_k)|\bm{s}_{k-1}=\bm{s} \right]$ is a vector consisting of expected values for each specific $\bm{s} \in \bm{S}$. For arbitrary decision $\bm{\delta}_{k-1}$, $\bm{\omega}_{k-1}^T \cdot \bigl[ \bm{g}(\bm{\delta}_{k-1}) + \mathop{\mathbb{E}}\left[\underline{V}_k(\bm{\omega}_k)|\bm{s}_{k-1}=\bm{s} \right] \bigl]$ is a linear function of $\bm{\omega}_{k-1}$. Finally, $\underline{V}_{k-1}(\bm{\omega}_{k-1})$ is the point-wise maximum of linear functions, which is a convex function. By induction, the reward-to-go function $\underline{V}_t(\omega_t), \forall t \in \{1,\dots,T\}$, is a convex function.
\end{proof}

The convexity of the reward-to-go function enables us to obtain an upper bound of $\underline{V}_t(\bm{\omega}_t)$. If we decompose the belief state as
\begin{equation}
\bm{\omega}_t = a_1 \bm{v}_1 + \dots + a_n \bm{v}_n,
\end{equation}
where $\bm{v}_i$s come from an arbitrary set of belief state vectors and $\sum_{i=1}^n a_i = 1$, then $\underline{V}_t(\bm{\omega}_t)$ satisfies the following inequality
\begin{equation}
\underline{V}_t(\bm{\omega}_t) \leq \overline{V}_t(\bm{\omega}_t) = a_1 \underline{V}_t(\bm{v}_1) + \dots + a_n \underline{V}_t(\bm{v}_n),
\label{reward_upper}
\end{equation}
where $\overline{V}_t(\bm{\omega}_t)$ represents an upper bound of $\underline{V}_t(\bm{\omega}_t)$. Therefore we can approximate the optimal reward-to-go function by the optimal value of the upper bound. The remaining goal is to select a set of belief state vectors and optimize each of them on the right-hand-side of \eqref{reward_upper}.

\subsection{Upper Bound of the Future Reward}
The previous approximation enables us to obtain an approximate value for the reward-to-go function $\underline{V}_t(\bm{\omega}_t)$ by a linear combination of $\underline{V}_t(\bm{v}_i)$. However, it remains difficult to compute $\underline{V}_t(\bm{v}_i)$. We observe that the reward-to-go function \eqref{approx_recur} consists of two parts: an immediate reward and a future reward. The computational complexity lies in the future reward, because of the uncountably infinite observation space. To resolve this issue, we approximate the future reward by its upper bound, denoted as $\tilde{V}_t(\bm{\omega}_t)$, which is obtained by assuming no error in future observations. It is obvious that decision making under observation errors yields lower rewards. With noiseless observations, the delayed observation only results in a finite number of possible belief vectors. Specifically, $\bm{v}_i$s are columns of $\bm{P}(\bm{\delta}_{t-2}) \cdot \bm{P}(\bm{\delta}_{t-1})$, because of the delay of two time slots. In this way, we have a finite number of states in the belief space and we can solve the problem via value iteration, with the following recursion: 
\begin{equation}
\tilde{V}_t (\bm{\omega}_t) = \max_{\bm{\delta}_t \in \bm{\Delta}_f}\Bigg[\bm{\omega}_t^T \bm{g}(\bm{\delta}_t)+\sum_{\bm{\omega}_{t+1}}P(\bm{\omega}_{t+1}|\bm{\omega}_t,\bm{\delta}_t)\tilde{V}_{t+1}(\bm{\omega}_{t+1})\Bigg],
\label{eq:V_tilde_C}
\end{equation}
which enables us to obtain upper bounds for the future rewards iteratively. The expected throughput $\tilde{V}_t(\bm{\omega}_t)$ can be decomposed as $\tilde{V}_{P,t}(\bm{\omega}_t) + \tilde{V}_{S,t}(\bm{\omega}_t)$ and the feasible set $\bm{\Delta}_f$ contains decisions satisfying $\tilde{V}_{P,t}(\bm{\omega}_t) \geq \beta \cdot V_{P,t}(\Psi_0)$.

With the set of belief vectors obtained from delayed observation and the approximated future rewards, we can optimize each individual term of the sum on the right-hand-side of \eqref{reward_upper} and the optimal decision for arbitrary belief state $\bm{\omega}_t$ can be constructed by a stochastic time-sharing approach. The pseudo-code for the proposed CCTS scheme is shown in Algorithm 1. The proposed scheme consists of two phases: offline planning and online decision-making. To analyze the computational complexity of both phases, consider $m$ possible system states and $d$ decision choices. For the offline planning phase, because the size of the basis set $\{v_1,\dots,v_n\}$ is $n = d^2 \cdot m$, the computational complexity of determining $\tilde{V}_t (\bm{\omega}_t)$ is $\mathcal{O}(d^3 m T)$. During online decision-making, maximizing over $d$ decisions leads to a computational complexity of only $\mathcal{O}(d)$.

\begin{algorithm}
\begin{algorithmic}[1]
\caption{The Centralized Cognitive Time Slot Scheduling Scheme}
\REQUIRE{state set $\bm{S}$, decision set $\bm{\Delta}$, basis set $\{\bm{v}_1, \dots, \bm{v}_n\}$}
\STATE Compute $\bm{P}(\bm{\delta})$, $\forall \bm{\delta}$, as described in Section \ref{SM-TP}.
\STATE Compute $\bm{g}(\bm{\delta})$, $\forall \bm{\delta}$, according to \eqref{eq:th_C}.
\HEADER{Offline Planning}
\FOR{t = T:-1:1}
\STATE Recursively compute $\tilde{V}_t(\bm{v}_i)$, $\forall i$, based on \eqref{eq:V_tilde_C}.
\ENDFOR
\ENDHEADER
\HEADER{Online Decision Making}
\FOR{t = 1:T}
\STATE Each SU takes action according to decision $\bm{\delta}_t$.
\STATE The CC obtains observation $\bm{y}_{t-1}$ and updates belief $\bm{w}_t$ according to \eqref{eq:belief_update_C}.
\STATE Find optimal decision $\bm{\delta}_{t+1}$ by maximizing $\tilde{V}_{t+1}(\bm{w}_{t+1}|\bm{w}_t)$ according to \eqref{eq:V_tilde_C}.
\STATE The CC broadcast decision $\bm{\delta}_{t+1}$ to all SUs.
\STATE SUs send observed $\bm{y}_{t}$ to the CC.
\ENDFOR
\ENDHEADER
\label{alg:CCTS}
\end{algorithmic}
\end{algorithm}

\section{The Optimal Decentralized Scheduling Problem} \label{ODSP}
In the centralized scheduling setup, with the help of a central controller, cooperation among SUs enables them to have a better understanding of the underlying system state and avoid mutual interference. However, because of the large propagation delay in UAC, centralized scheduling faces huge communication overhead. Furthermore, deployment and maintenance of central controllers remains difficult. Therefore, it is desirable to develop decentralized scheduling schemes for UAN. In this section, we formulate the optimal interference-constrained scheduling problem for a \textit{decentralized} SU network. Our goal is to maximize the expected end-to-end throughput of the overall system over a finite horizon of $T$ time slots, in a decentralized manner.

\subsection{Problem Formulation}
Over the total $T$ time slots, each SU should make independent decisions based solely on its previous observations and decisions. Denote a policy for SU $i$ as $\Psi^i$, which is a mapping from its histories to a decision,
\begin{equation}
\delta_t^i = \Psi^i(\delta_0^i, \dots,\delta_{t-1}^i, y_0^i, \dots, y_{t-1}^i),
\end{equation}
then the optimization problem can be formulated as follows
\begin{align}
&\text{Maximize:} & &V(\Psi^1,\dots,\Psi^{N_S}) = \mathbb{E}\left\{\sum_{t=1}^T g(\bm{s}_t, \bm{\delta}_t) \bigg| \Psi^1,\dots,\Psi^{N_S} \right\} \\
&\text{Subject to:} & &\mathbb{E}\left\{\sum_{t=1}^T g_P(\bm{s}_t, \bm{\delta}_t) \bigg| \Psi^1,\dots,\Psi^{N_S} \right\} \geq \beta\cdot V_P(\Psi_0),
\end{align}
where the expectation is over states and observations. Each SU only has observability over a limited portion of the overall system and has to make independent decisions. There are two main differences between the centralized scheduling problem and the decentralized problem. First, in the decentralized case, individual observation data and decision strategy cannot be shared among SUs, so that each SU has to make decision without cooperation with other SUs. Second, since the SUs do not communicate with any coordinator in the decentralized case, the delay is reduced to one time slot. This problem can be viewed as a decentralized POMDP (Dec-POMDP), finding the optimal solution for a finite-horizon Dec-POMDP has the time complexity of NEXP-complete~\cite[Chapter~7]{kochenderfer2015decision}. Therefore we formulate an approximation of the original problem.

\subsection{The Approximated Problem} \label{ODSP-AP}
In the approximated problem, instead of end-to-end throughput, we maximize the per hop throughput for each SU hop. Thus, the original optimization problem is divided into $N_S$ sub-problems, one for each SU hop. Notice that each SU may not be able to sense activities of the entire system, but as a limited portion of it.

Let $C_i$ denote the set of PUs within the one slot delay region of the $i$-th SU. Then denote the joint state of PUs in $C_i$ as $\bm{s}_{i,t}$, which is a portion of the entire PU state $\bm{s}_t$. 
\begin{equation}
\bm{s}_{i,t} \triangleq (\dots, s_{i,t}^j, \dots)^T \in \bm{S}_i \equiv \{0,1\}^{\lVert C_i \rVert}, \quad j \in C_i,
\end{equation}
where $s_{i,t}^j$ is the state of the $j$-th PU hop in time slot $t$. The temporal evolution of $\bm{s}_{i,t}$ can be described by a Markov Chain, the only undetermined parameter is the traffic density at the first PU node in $C_i$.
\begin{rem}
Although the traffic density coefficients at each PU node are determined by the scheduling policies of the SUs, they are upper-bounded by $\alpha_1$ and $\alpha_2$, which are the traffic density coefficients at the PU transmitter. Therefore we use $\alpha_1$ and $\alpha_2$ as the traffic density coefficients to derive SU policies so that the interference to PUs is bounded below.
\end{rem}
Thus we can calculate the transition matrix for the behavior of PUs in $C_i$ using the traffic density coefficients $\alpha_1$, $\alpha_2$ and the packet loss probability,
\begin{equation}
p^j_P(\bm{s}_{i,t}, \delta_t^i) \approx 1 - \prod_{k=1}^{L_P}\left[1 - P_b^k(\bm{s}_{i,t},\delta_t^i)\right].
\end{equation}
The resulting transition matrix is denoted as $\bm{P}_i(\delta_t^i)$. For the behavior of other SUs, no prior knowledge is known in the process of policy design. Periodic scheduling with spatial reuse has been investigated for interference alleviation for a single class of users \cite{zhang2010analysis,chen2007delay}. We utilize this approach to alleviate interference among SUs, and focus on managing interference to PUs.

Each SU performs channel sensing independently to collect noisy measurements. The observation model for the $i$-th SU in time slot $t$ is defined in \eqref{eq:obs_model}. The information vector in time slot $t$, denoted as $\bm{I}_t^i$, can be expressed as
\begin{equation}
\bm{I}_t^i = (y_1^i,\dots,y_{t-1}^i,\delta_1^i, \dots,\delta_{t-1}^i),
\end{equation}
which is of expanding dimension. However, for control purposes, we can use the probability distribution of state $\bm{s}_{i,t}$, denoted by $\bm{\omega}_t^i$, as a sufficient statistic. $\bm{\omega}_t^i$ is also known as belief state and is defined as
\begin{equation}
\bm{\omega}_t^i = [\dots, \omega_t^i(s), \dots]^T; \quad \omega_t^i(s) \triangleq \mathbb{P}(\bm{s}_{i,t} = s|\bm{I}_t^i), s \in \bm{S}_i.
\end{equation} 

Given observation $y_t^i$, the $i$-th SU can update its belief state via Bayes' Rule as follows
\begin{equation}
\omega_t^i(s) \triangleq \mathbb{P}(\bm{s}_{i,t} = \bm{s}|\bm{I}_t^i) = \frac{\left[\mathbb{P}(y_t^i|\bm{s}_{i,t}=\bm{s})\right][(\bm{\omega}_{t-1}^i)^T\bm{P}_i(\delta_{t-1}^i)_{\bm{s}}]}{\sum_{\bm{s}' \in \bm{S}_i}\left[\mathbb{P}(y_t^i|\bm{s}_{i,t}=\bm{s}')\right][(\bm{\omega}_{t-1}^i)^T \bm{P}_i(\delta_{t-1}^i)_{\bm{s}'}]},
\label{eq:belief_update_D}
\end{equation}
where $\bm{P}(\delta_t^i)_{\bm{s}}$ represents the column of the transition matrix $\bm{P}(\delta_t^i)$ corresponding to state $\bm{s}$ and $\mathbb{P}(y_t^i|\bm{s}_{i,t}=\bm{s})$ is the probability distribution function of a Gaussian distribution evaluated at $y_t^i$. Further, we denote the belief state update rule as $\Phi_i(\cdot)$, then $\bm{\omega}_t^i = \Phi_i(\bm{\omega}_{t-1}^i,\delta_{t-1}^i,y_t^i)$.

Notice that because of the large propagation delay, the $i$-th SU can only obtain $\bm{\omega}_t^i$ at the end of time slot $t$. Therefore it has to make decision for time slot $t$ only based on the belief $\bm{\omega}_{t-1}^i$ calculated from previous observation.

Now we can formulate the optimal interference-constrained scheduling problem. Denote the one slot throughput of PUs in $C_i$ as $g_P^i(\cdot)$, which is defined as follows
\begin{equation}
g_P^i(\bm{s}_{i,t}, \delta_t^i) \triangleq L_P \cdot s_{i,t}^r \cdot p_P^r(\bm{s}_{i,t}, \delta_t^i),
\end{equation}
where $r$ is the PU index of the end receiver in $C_i$. The one slot throughput for the $i$-th SU hop is defined as:
\begin{equation}
g_S^i(\bm{s}_{i,t}, \delta_t^i) \triangleq L_S^i \cdot \delta_t^i \cdot p_S^i(\bm{s}_{i,t}, \delta_t^i),
\end{equation}
where $L_S^i$ is the packet size of the $i$-th SU hop. The total throughput in slot $t$ is denoted as $g^i(\cdot)$, and defined as
\begin{equation}
g^i(\bm{s}_{i,t}, \delta_t^i) \triangleq g_P^i(\bm{s}_{i,t}, \delta_t^i) + g_S^i(\bm{s}_{i,t}, \delta_t^i).
\label{eq:th_D}
\end{equation}

The PU throughput degradation constraint of the original optimization problem is not applicable because each SU has limited observability such that the overall PU performance can not be guaranteed. Therefore we approximate the original end-to-end PU throughput constraint by the constraint on the throughput degradation of PUs in $C_i$. Denote the expected accumulated throughput of both PUs and the $i$-th SU hop over a finite horizon as $V^i$, then the optimization problem can be formulated as follows
\begin{align}
&\text{Maximize:} & &V^i(\Psi^i) = \mathbb{E}\left\{\sum_{t=1}^T g^i(\bm{s}_{i,t}, \delta_t^i) \bigg| \Psi^i \right\} \\
&\text{Subject to:} & & \mathbb{E}\left\{\sum_{t=1}^T g^i_P(\bm{s}_{i,t}, \delta_t^i) \bigg| \Psi^i \right\} \geq \bar{\beta} \cdot \mathbb{E}\left\{\sum_{t=1}^T g^i_P(\bm{s}_{i,t}, 0) \right\},  \\
& & &\delta_t^i = 0, \quad \forall t \not\equiv i \pmod{R},
\end{align}
where $\bar{\beta}$ is the local degradation coefficient, which is related to $\beta$ as $\bar{\beta} = \beta^{(1/N_S)}$. The second constraint states that each SU is given an opportunity to access the channel periodically, with a spatial reuse factor $R$. For this paper, we use $R = 3$, which is appropriate for mutual interference alleviation and avoidance of the hidden terminal problem. Using the belief state $\bm{\omega}_t^i$, we can rewrite the objective function as
\begin{equation}
V^i(\Psi^i) = \mathbb{E}\left\{\sum_{t=1}^T (\bm{\omega}_t^i)^T \cdot \bm{g}^i(\delta_t^i) \bigg| \Psi^i \right\}
\end{equation}
where $\bm{g}^i(\delta_t^i)$ is the vector that consists of total throughput at each state.

\subsection{Optimal Packet Size}
The significant propagation delay in underwater acoustic channels provides unique opportunity for scheduling in the time domain. The desired signal and potentially interfering signal, even if transmitted at the same time, may not overlap at the receiver side. This phenomenon is illustrated in Fig. \ref{fig:packet_diagram}, where a SU can avoid interference to a certain PU hop by properly adjusting its packet size. Thus, the interference between any SU and its surrounding PUs is determined by the packet size of that SU hop. As each SU has distinct geometric relationship with the PU network, we expect the optimal packet size for each SU hop to be different. 

\begin{figure}[!t]
\centering
\subfigure[Normal SU packet size, collision occurs.]{
\includegraphics[width=0.48\linewidth]{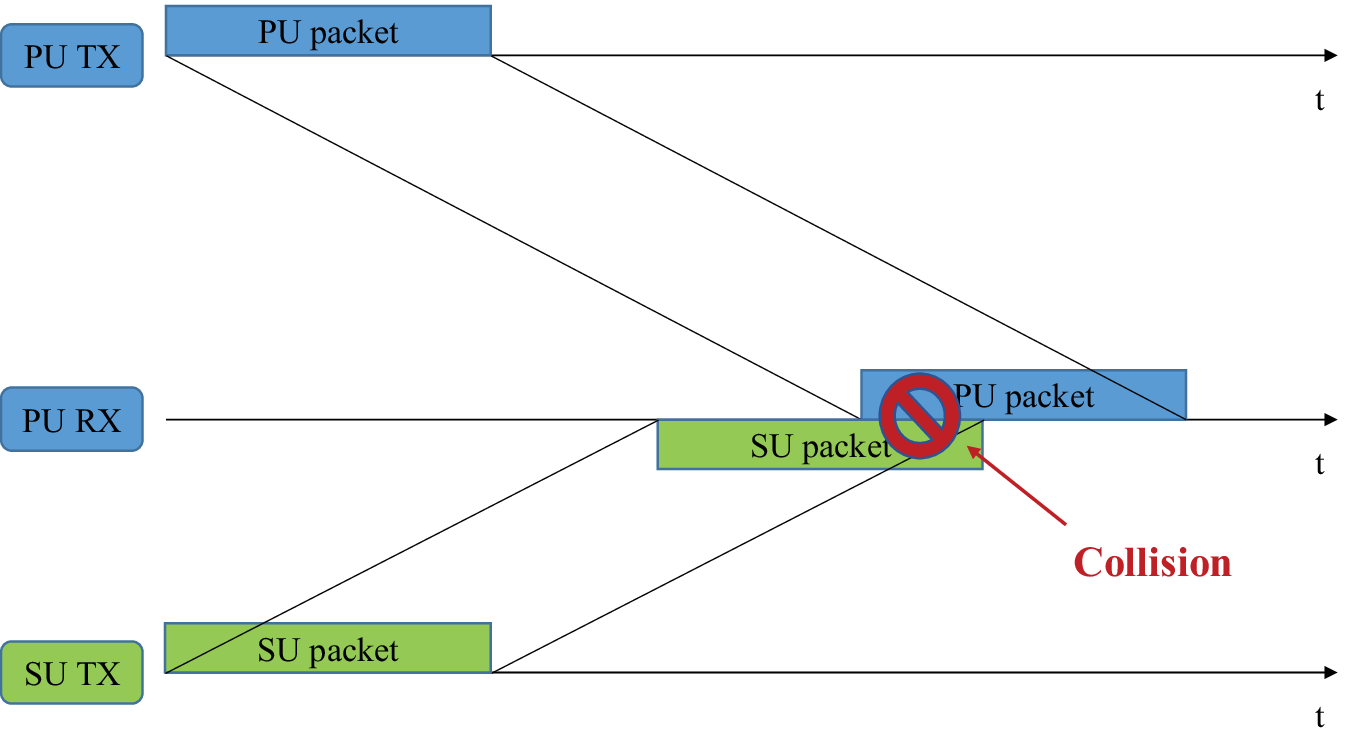}}
\subfigure[Reduced SU packet size, interference-free.]{
\includegraphics[width=0.48\linewidth]{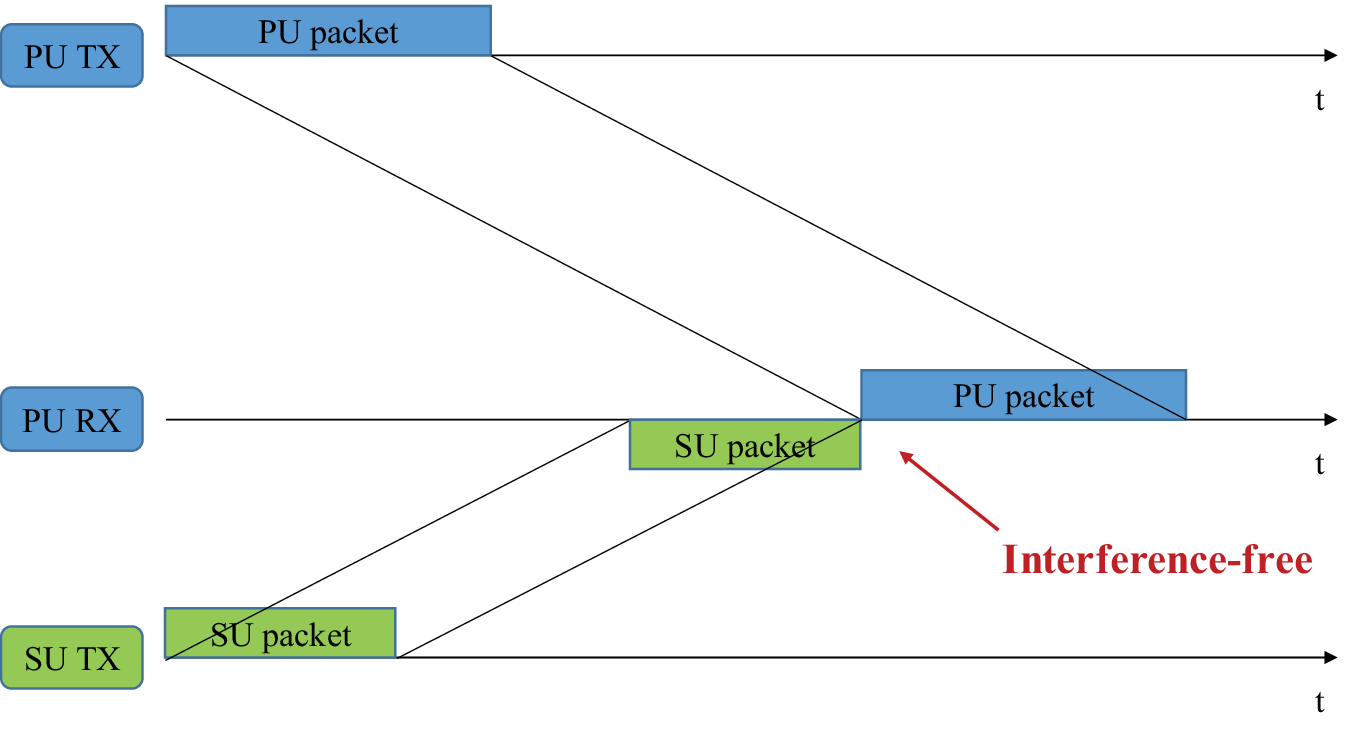}}
\caption{Scenario analysis for the influence of packet size.}
\label{fig:packet_diagram} 
\end{figure}

As we stated previously, the optimal decentralized scheduling problem has a time complexity of NEXP-complete. Jointly optimizing packet sizes and scheduling of SUs further complicates the problem, thus the joint problem is not computationally feasible. We propose to determine packet sizes by assuming a simple decision strategy for SUs, before solving the scheduling problem. Assume the $i$-th SU transmits data with probability $q_k^i$, given previous state $k$, such that the decision of the $i$-th SU is a Bernoulli random variable with parameter $q_k^i$, $\delta^i(k) \sim \text{Ber}(q_k^i)$. The probability distribution of states is affected by the SU activity, which is unknown. We approximate the true state distribution by the steady state distribution without SU interference, denoted by $\bm{\pi}^i$. The optimization problem is formulated as: 
\begin{align}
&\text{Maximize} & &f_i(L_S^i, q_k^i) = (\bm{\pi}^i)^T \bm{P}_i \cdot \mathbb{E}\left[\bm{g}^i(\delta^i(k))\right] \\
&\text{Subject to} & &(\bm{\pi}^i)^T \bm{P}_i \cdot \mathbb{E}\left[\bm{g}_P^i(\delta^i(k))\right] \geq \beta \cdot (\bm{\pi}^i)^T \bm{P}_i \cdot \bm{g}_P^i(0),
\end{align}
where $\mathbb{E}\left[\bm{g}^i(\delta^i(k))\right]$ and $\mathbb{E}\left[\bm{g}_P^i(\delta^i(k))\right]$ can be calculated as
\begin{equation}
\mathbb{E}\left[\bm{g}^i(\delta^i(k))\right] = q_k^i \cdot \bm{g}^i(1) + (1 - q_k^i) \cdot \bm{g}^i(0),
\end{equation}
\begin{equation}
\mathbb{E}\left[\bm{g}_P^i(\delta^i(k))\right] = q_k^i \cdot \bm{g}_P^i(1) + (1 - q_k^i) \cdot \bm{g}_P^i(0).
\end{equation}
Notice that for a fixed packet length $L_S^i$, $\bm{g}^i(0)$, $\bm{g}^i(1)$, $\bm{g}_P^i(0)$ and $\bm{g}_P^i(1)$ are constant numbers, thus the problem is a linear program. 

Another important fact is that, with increasing packet size, the expected throughput drops drastically as the packet starts to overlap with packets of other users. Define the critical size $L_{c,i}^P(i,j)$ to be the maximum packet size such that the $j$-th PU is not interfered by the $i$-th SU. Similarly, define critical size $L_{c,i}^S(i,j)$ to be the maximum packet size such that the $i$-th SU is not interfered by the $j$-th PU. The set of critical sizes is defined as
\begin{equation}
L_{C,i} \triangleq \left\{L_C^P(i,j)| \forall j\right\} \cup \left\{L_C^S(i,j)| \forall j\right\}.
\end{equation}
For each element of $L_{c,i}$, we optimize $q_k^i$s by linear programming. Then we select the optimal value $\left(L_S^i\right)^*$, according to the following equation
\begin{equation}
\left(L_S^i\right)^* = \underset{L_S^i \in L_c}{\arg\max} \quad f_i\left(L_S^i, \left(q_k^i\right)^*\right).
\end{equation}

\section{Decentralized Cognitive Time Slot Scheduling} \label{DCTS}
In the previous section, we argued that optimal decentralized scheduling is computationally infeasible, thus we divide the original problem into several approximated sub-problems. Although we transforms one Dec-POMDP problem into several POMDP problems by decoupling, they are still computationally prohibitive. In this section, we present a low-cost approximation scheme that can be implemented and run in real-time. We denote this method as the \textit{Decentralized Cognitive Time Slot Scheduling (DCTS)} scheme. By exploring the unique characteristics of the system model and properties of the reward-to-go function presented in Section \ref{CCTS}, we develop a computationally efficient scheme, enabled by a sequence of approximations. As weill see in the numerical results, the proposed scheme despite being based on these approximations, offers strong performance.

Denote the reward-to-go function as $V_t^i(\cdot)$.
Firstly, as the global constraint greatly complicates computation, we approximate it by a set of local constraints,
\begin{equation}
\mathbb{E}\left\{\sum_{\tau=t}^T g^i_P(\bm{s}_{i,t}, \delta_t^i) \bigg| \Psi^i \right\} \geq \bar{\beta} \cdot V_{P,t}^i(\Psi_0), \quad \forall t \in \{1,2,\dots,T\},
\end{equation}
where $V_{P,t}^i(\Psi_0) = \mathbb{E}\left\{\sum_{\tau=t}^T g^i_P(\bm{s}_{i,t}, 0)\right\}$. In this way, we obtain a lower bound of $V_t^i$, denoted as $\underline{V}_t^i$.

Secondly, as the reward-to-go function is convex and piece-wise linear, we can obtain an upper bound of it. If we decompose the belief state as
\begin{equation}
\bm{\omega}_t^i = a_1 \bm{v}_1 + \dots + a_n \bm{v}_n,
\end{equation}
where $\bm{v}_i$s come from an arbitrary set of belief state vectors and $\sum_{i=1}^n a_i = 1$. Then $\underline{V}_t^i(\bm{\omega}_t^i)$ satisfies the following inequality
\begin{equation}
\underline{V}_t^i(\bm{\omega}_t^i) \leq \overline{V}_t^i(\bm{\omega}_t^i) = a_1 \underline{V}_t^i(\bm{v}_1) + \dots + a_n \underline{V}_t^i(\bm{v}_n),
\end{equation}
where $\overline{V}_t^i(\bm{\omega}_t^i)$ represents an upper bound of $\underline{V}_t^i(\bm{\omega}_t^i)$. Therefore we can approximate the optimal reward-to-go function by the optimal value of the upper bound.

Thirdly, we approximate the future reward by its upper bound, denoted as $\tilde{V}_t^i(\bm{\omega}_t^i)$, by assuming no errors in future observations. It is obvious that decision making under observation error yields lower rewards. With noiseless observation, delayed observation only results in a finite number of possible belief vectors (\ie, $\bm{v}_i$s are columns of $\bm{P}_i(\delta^i)$). In this way, we have finite states in the belief space and we can solve the problem by the value iteration method, with the following recursion
\begin{equation}
\tilde{V}_t^i (\bm{\omega}_t^i) = \max_{\delta_t^i \in \Delta_f^i}\Bigg[(\bm{\omega}_t^i)^T \bm{g}^i(\delta_t^i)+\sum_{\bm{\omega}_{t+1}^i}P(\bm{\omega}_{t+1}^i|\bm{\omega}_t^i,\delta_t^i)\tilde{V}_t^i(\bm{\omega}_{t+1}^i)\Bigg],
\label{eq:V_tilde_D}
\end{equation}
where $\Delta_f^i$ is the set of feasible decisions.

Based on the aforementioned approximations, the decentralized scheduling strategy has a threshold structure. Specifically, for SU $i$, transmission is a better choice in terms of total throughput if
\begin{equation}
(\bm{\omega}_t^i)^T\bm{g}^i(1)+\sum_{\bm{\omega}_{t+1}^i}P(\bm{\omega}_{t+1}^i|\bm{\omega}_t^i,1)\tilde{V}_t^i(\bm{\omega}_{t+1}^i) \geq (\bm{\omega}_t^i)^T \bm{g}^i(0)+\sum_{\bm{\omega}_{t+1}^i}P(\bm{\omega}_{t+1}^i|\bm{\omega}_t^i,0)\tilde{V}_t^i(\bm{\omega}_{t+1}^i),
\end{equation}
which is equivalent to:
\begin{equation}
(\bm{\omega}_t^i)^T \left[\bm{g}^i(1)-\bm{g}^i(0)\right] \geq \sum_{\bm{\omega}_{t+1}^i} \left[P(\bm{\omega}_{t+1}^i|\bm{\omega}_t^i,0) - P(\bm{\omega}_{t+1}^i|\bm{\omega}_t^i,1) \right] \cdot \tilde{V}_t^i(\bm{\omega}_{t+1}^i).
\label{eq:threshold_D}
\end{equation}
The PU throughput degradation constraint can also be expressed as a threshold structure:
\begin{equation}
(\bm{\omega}_t^i)^T\bm{g}_P^i(\delta_t^i)+\sum_{\bm{\omega}_{t+1}^i}P(\bm{\omega}_{t+1}^i|\bm{\omega}_t^i,\delta_t^i)\tilde{V}_{P,t}^i(\bm{\omega}_{t+1}^i) \geq \bar{\beta} \cdot V_{P,t}^i(\Psi_0),
\label{eq:threshold_constr_D}
\end{equation}
which is used to determine the feasible set $\Delta_f^i$.

Therefore, we can compute the approximated reward-to-go values $\tilde{V}_t^i(\bm{\omega}_{t+1}^i)$ in advance. In run-time, the decision can be made by making comparison to thresholds. The pseudo-code for the proposed DCTS scheme for SU $i$ is shown in Algorithm 2. The DCTS scheme also consists of two phases: offline planning and online decision-making. We analyze the computational complexity of the DCTS scheme in terms of $m_i$ possible system states and $d_i$ decision choices, where $d_i = 2$ because of the binary decision. For the offline planning phase, the computational complexity of determining $\tilde{V}_t^i (\bm{\omega}_t^i)$ is $\mathcal{O}(d_i^2 m_i T)$. During online decision-making, checking a threshold condition leads to a computational complexity of only $\mathcal{O}(d_i)$.

\begin{algorithm}
\begin{algorithmic}[1]
\caption{The Decentralized Cognitive Time Slot Scheduling Scheme}
\REQUIRE{state set $\bm{S}_i$, decision set $\Delta_i$, basis set $\{\bm{v}_1, \dots, \bm{v}_n\}$}
\STATE Compute $\bm{P}_i(\delta^i)$, $\forall \delta^i$, as described in Section \ref{ODSP-AP}.
\STATE Compute $\bm{g}^i(\delta^i)$, $\forall \delta^i$, according to \eqref{eq:th_D}.
\HEADER{Offline Planning}
\FOR{t = T:-1:1}
\STATE Recursively compute $\tilde{V}_t^i(\bm{v}_i)$, $\forall i$, based on \eqref{eq:V_tilde_D}.
\ENDFOR
\ENDHEADER
\HEADER{Online Decision Making}
\FOR{t = 1:T}
\STATE Update belief $\bm{w}_t^i$ based on observation $y_{t-1}^i$ according to \eqref{eq:belief_update_D}.
\STATE Find optimal decision $\delta_t^i$ by checking the threshold conditions \eqref{eq:threshold_D} and \eqref{eq:threshold_constr_D}.
\STATE Collect observational data $y_t^i$.
\ENDFOR
\ENDHEADER
\label{alg:DCTS}
\end{algorithmic}
\end{algorithm}

\section{Numerical Results} \label{NR}
In this section, we numerically evaluate the performance of the CCTS scheme presented in Section \ref{CCTS} and the DCTS scheme presented in Section \ref{DCTS}. The experiment is based on the acoustic multipath channel simulator developed in \cite{qarabaqi2013statistical}, which takes into account small-scale and large-scale channel variations.

We consider a CM-UAN with $N_P = 4$ PU hops and $N_S = 4$ SU hops, with a depth of $100$ \unit{m}. The end-to-end distances for both the PU and SU multi-hop networks are set as $10$ \unit{km}, thus the per hop distance for both PUs and SUs is $2.5$ \unit{km}. The crossing-lines topology, where significant interference exists between PUs and SUs is selected to facilitate the controlled changing of the topology. Although we focus on a specific topology for simulation, the proposed scheme is not dependent on a particular topology, under the assumption of multi-hop relaying.

The carrier frequency is $f_c = 32$ kHz with $4$ \unit{kHz} bandwidth and the transmit power for both PUs and SUs is $130$ \unit{dB} re \unit{\micro\pascal}. For PUs, packet size is $L_P = 1500$ bytes and the transmission rate is $10$ \unit{kbps} for both PUs and SUs. Taking both propagation and transmission delays into account, the time slot length is about $3$ \unit{s}. The performance is evaluated via simulations over a finite horizon of $T = 1000$ slots, and each Monte Carlo run is repeated $100$ times.

We conduct the experiment on a computer with a $4.1$ \unit{G\hertz} Intel-i5 CPU and a $16$ \unit{GB} memory. For the proposed DCTS scheme, the offline planning takes $1.3$ \unit{s} while the online decision-making takes $0.18$ \unit{ms} per time step, on average. The running time result confirms that the proposed scheme is computationally efficient and applicable in real-world scenarios.

\subsection{Comparison Algorithms}
Based on the aforementioned settings, we compare proposed algorithms with three comparison schemes, which can be classified into time division multiplexing (TDM) and frequency division multiplexing (FDM). We illustrate the difference in their operating mechanism in Fig. \ref{fig:T_vs_F}.

\begin{figure}[!t]
\centering
\subfigure[Time division multiplexing.]{
\includegraphics[width=0.48\linewidth]{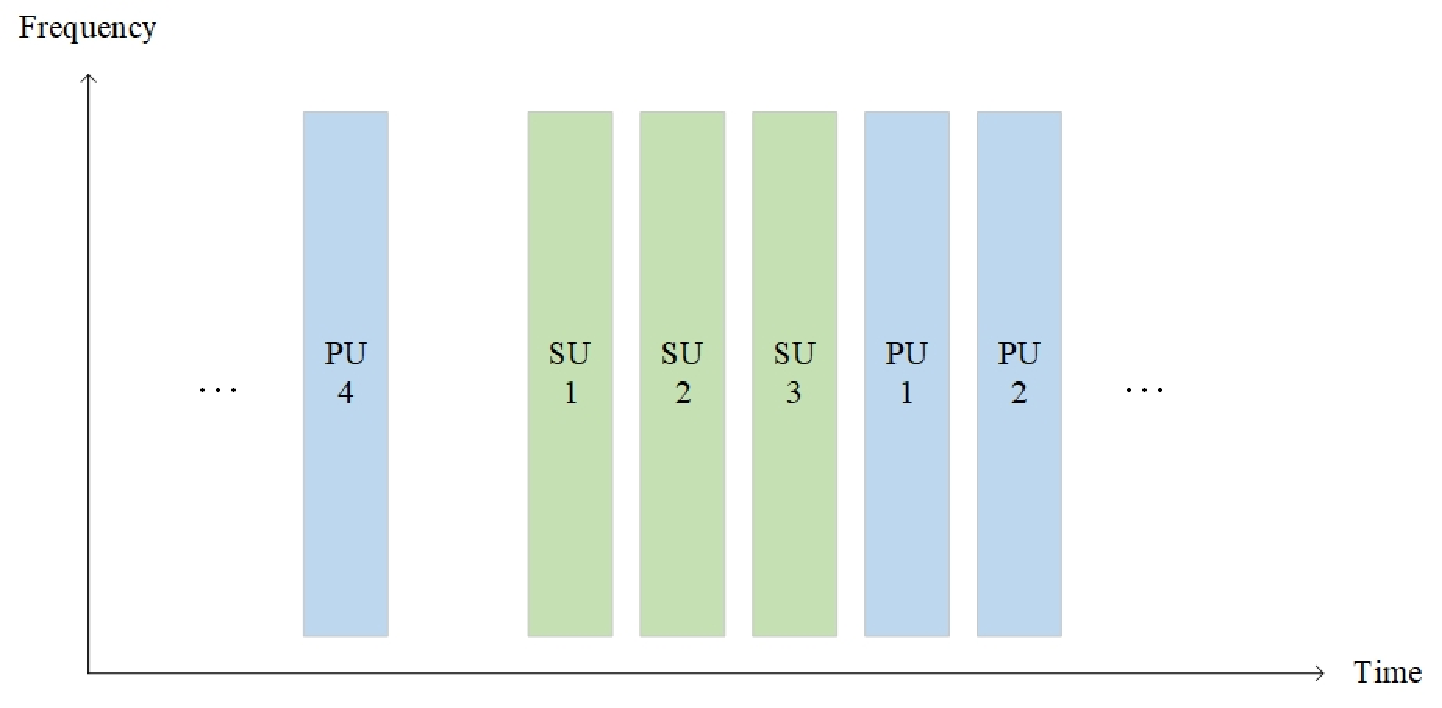}}
\subfigure[Frequency division multiplexing.]{
\includegraphics[width=0.48\linewidth]{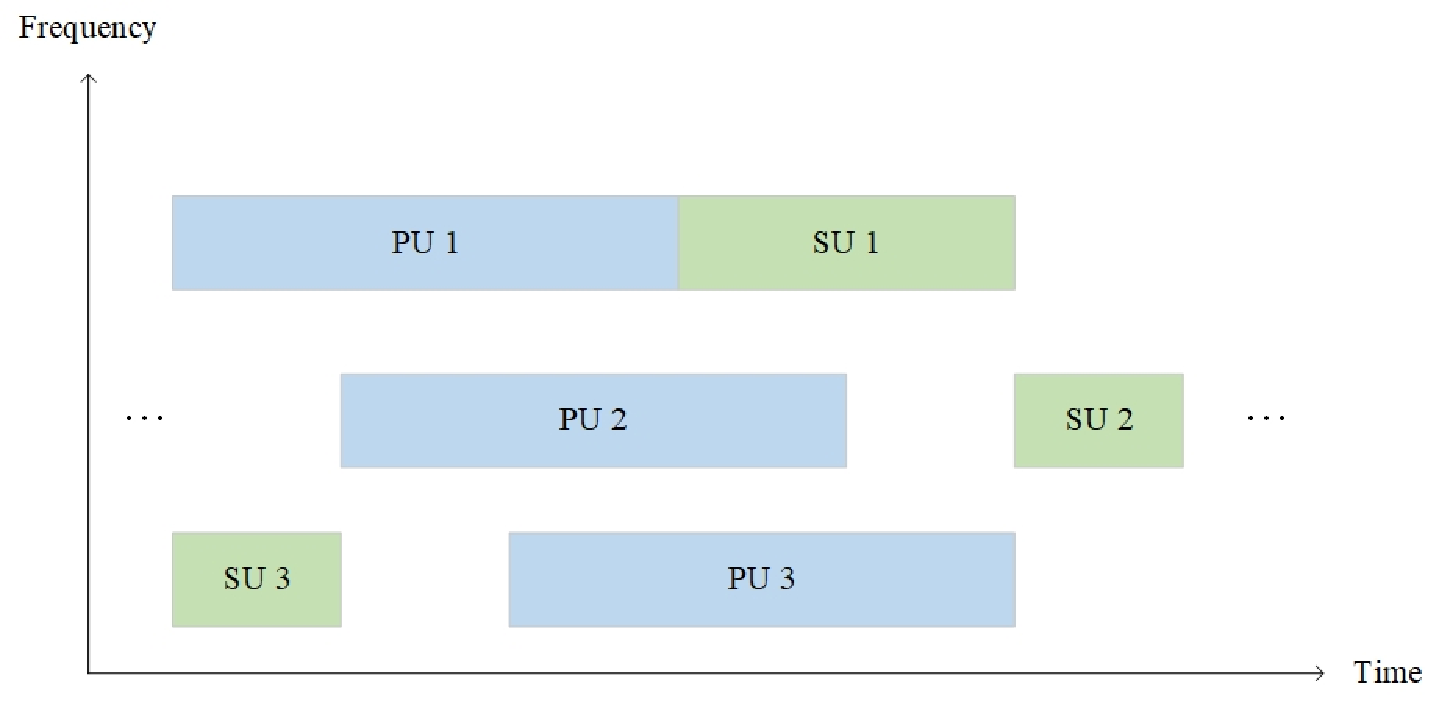}}
\caption{Illustration of time division multiplexing scheme and frequency division multiplexing scheme in cognitive UAN.}
\label{fig:T_vs_F}
\end{figure}

The conventional time-division multiplexing (C-TDM) algorithm \cite[Chapter~6]{xiao2008cognitive} is the first comparison scheme. In order to take the constraint into account, SUs have the opportunity to make access decisions with probability $1-\bar{\beta}$. Based on channel measurements, SUs make decisions to transmit if the estimated channel occupancy probability is lower than a fixed threshold. Here we select the fixed threshold to be $0.5$ such that the optimal decision maximizes the probability of successful transmission. The decision rule can be expressed as $\delta_t^i = X \cdot Y^i_t$ where $X \sim \text{Bernoulli}(1-\bar{\beta})$ and $Y_t^i$ is defined as follows
\begin{equation}
Y_t^i = \begin{cases}
1, \text{ if } \sum_{\bm{s}} \omega_t^i(\bm{s}) \cdot \mathbb{I}(\bm{s} \neq \bm{0}) \leq 0.5 \\
0, \text{ otherwise}
\end{cases},
\end{equation}
where $\mathbb{I}(\cdot)$ is the indicator function and $\sum_{\bm{s}} \omega_t^i(\bm{s}) \cdot \mathbb{I}(\bm{s} \neq \bm{0})$ represents the channel occupancy probability.

Inspired by \cite{zeng2017distributed}, we consider another time-domain scheduling scheme based on interference alignment (IA). Since the proposed algorithm in \cite{zeng2017distributed} is not designed for cognitive networks with coexisting PUs and SUs, we develop an adapted IA algorithm for more fair comparison. To achieve interference-free communication, the PUs periodically send their schedule to the SUs, with one slot communication overhead. With scheduling information of PUs, each SU makes decisions by aligning interference patterns. The interference among SUs is alleviated by a random access protocol, such that each SU becomes active with a particular probability. Notice that the one slot overhead we assumed for information exchange may be a lower bound on the actual overhead needed.

Apart from scheduling in the time domain, we also consider scheduling schemes in the frequency domain. The conventional frequency division multiplexing (C-FDM) algorithm is the third comparison scheme. The total bandwidth is divided into three sub-channels, each with $1.2$ \unit{kHz} bandwidth and center frequency of $30.6$ \unit{kHz}, $32$ \unit{kHz}, $33.4$ \unit{kHz}, respectively. Guard bands of $0.2$ \unit{kHz} are inserted between sub-channels to prevent frequency spreading. PUs and SUs are assigned to different sub-channels, so there is no common channel usage among every three consecutive nodes. The transmission rate of each sub-channel is set as $3$ \unit{kbps} proportionally, and the SU packet sizes are reduced accordingly to fit the time slot length. To guarantee PU throughput, SUs are given opportunities to make access decisions with probability $1-\bar{\beta}$. Given the opportunity, each SU makes decisions independently about transmission based on channel measurements. Notice that each SU only cares about the occupation probability of the assigned channel. The decision rule can be expressed as $\delta_t^i = X \cdot Z^i_t$ and $Z_t^i$ is defined as follows 
\begin{equation}
Z_t^i = \begin{cases}
1, \text{ if } \sum_{\bm{s}} \omega_t^i(\bm{s}) \cdot \mathbb{I}\left(h_i(\bm{s}) \neq \bm{0}\right) \leq 0.5 \\
0, \text{ otherwise}
\end{cases},
\end{equation}
where $h_i(\bm{s})$ represents a portion of the state vector $\bm{s}$, consisting of states of the PUs that share the same channel with the $i$-th SU hop.

In Section \ref{DCTS}, we propose the DCTS scheme, which schedules the SUs by time division multiplexing. Now, we consider a variation of the proposed scheme such that the SUs are scheduled by frequency division multiplexing, denoted as DCTS-FDM. The frequency division setting and parameters are the same as what we state previously. We formulate the optimization problem as follows
\begin{align}
&\text{Maximize:} & &V^i(\Psi^i) = \mathbb{E}\left\{\sum_{t=1}^T g^i(\bm{s}_{i,t}, \delta_t^i) \bigg| \Psi^i \right\} \\
&\text{Subject to:} & & \mathbb{E}\left\{\sum_{t=1}^T g^i_P(\bm{s}_{i,t}, \delta_t^i) \bigg| \Psi^i \right\} \geq \bar{\beta} \cdot \mathbb{E}\left\{\sum_{t=1}^T g^i_P(\bm{s}_{i,t}, 0) \right\}.
\end{align}
The constraint for periodic opportunity scheduling is not needed, because mutual interference for SUs is resolved by frequency division. Although the objective function and the constraint for PU performance have the same form as before, parameters and state transition probabilities are different in the FDM settings. The same set of approximations as the DCTS scheme is applied to solve this optimization problem.

\subsection{Centralized Scheduling}
Figure \ref{fig:th_alpha} plots both PU throughput and total (PU + SU) throughput as a function of the traffic density coefficient $\alpha_2$, with $\alpha_1 = \alpha_2/4$. Here, we fix the PU throughput degradation coefficient $\beta$ to be $0.8$ and provide the curve depicting this constraint (PU constraint). We notice that both the CCTS and the C-TDM schemes satisfy the PU throughput constraint, while the CCTS scheme is able to get closer to the actual constraint of PU throughput. Considering the case with no SU transmission as a baseline, the proposed CCTS scheme provides a $140.9\%$ average improvement in total throughput. In comparison, the C-TDM scheme only provides an average improvement of $59.6\%$.

\begin{figure}
\centering
\includegraphics[width=0.7\linewidth]{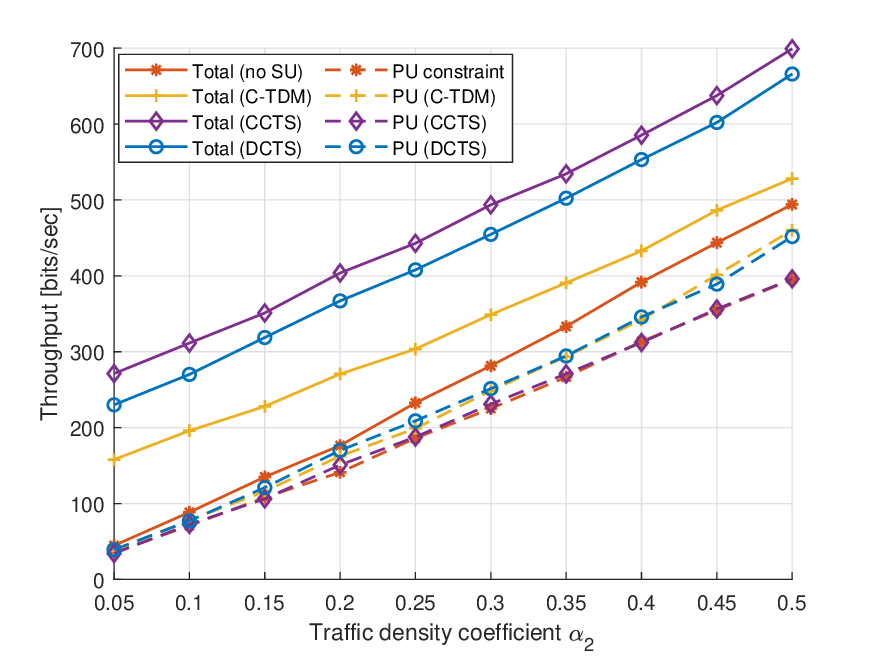}
\caption{Throughput comparison between CCTS, DCTS and C-TDM schemes, with varying traffic density coefficient $\alpha_2$, $\beta = 0.8$.}
\label{fig:th_alpha}
\end{figure}

Figure \ref{fig:th_beta} plots throughput as a function of PU throughput degradation coefficient $\beta$ where the traffic density coefficients are fixed as $\alpha_1 = 0.05$, $\alpha_2 = 0.2$. We observe an average improvement of $142.6\%$ in total throughput by the CCTS scheme, compared with the case with no SU transmission, while the C-TDM scheme provides a $65.8\%$ improvement on average. Additionally, notice that the system throughput of both schemes decreases as $\beta$ increases to larger values. When $\beta = 1$, no throughput improvement is provided by both schemes as no interference is allowed.

\begin{figure}
\centering
\includegraphics[width=0.7\linewidth]{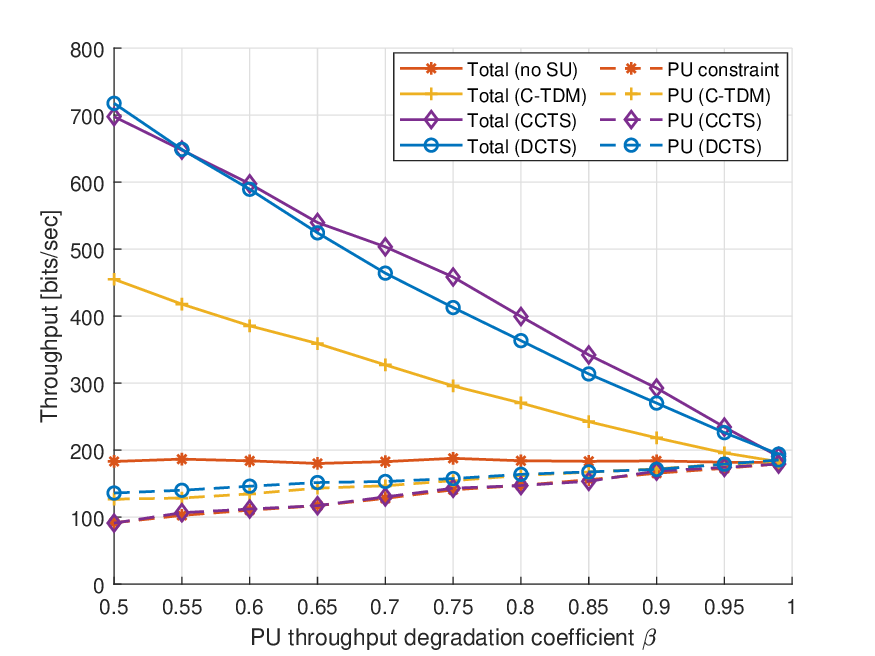}
\caption{Throughput comparison between CCTS, DCTS and C-TDM schemes, with varying degradation coefficient $\beta$, $\alpha_1 = 0.05$, $\alpha_2 = 0.2$.}
\label{fig:th_beta}
\end{figure}

\subsection{Decentralized Scheduling}
In Fig. \ref{fig:th_alpha}, notice that the DCTS scheme satisfies the PU throughput constraint but is farther away from the actual constraint when compared with the CCTS scheme. The reason is that, with the cooperation of all SUs, the central controller has a comprehensive understanding of PU behavior so that decisions can be made without conjecture. The DCTS scheme provides an average improvement of $116.4\%$ in total throughput, compared with the case with no SU transmission, which is much larger than the improvement achieved by the C-TDM scheme.

As depicted in Fig. \ref{fig:th_beta}, apart from satisfying the PU throughput constraint, the DCTS scheme provides $133.8\%$ improvement on average, compared with the case with no SU transmission, in terms of total throughput. The improvement level provided by DCTS and CCTS schemes are close, while much larger than the C-TDM scheme.

Notice that the achieved performance of the centralized scheme, CCTS, is under optimistic assumptions, including the existence of a CC, a separate band for the control signal, and a control overhead of only one time slot.  However, the decentralized scheme still offers performance close to the idealized centralized scheme's upper bound.  DCTS also offers lower computational complexity.

\subsection{Packet Size}
The performance gain provided by optimizing SU packet sizes is evaluated versus traffic density in Fig. \ref{fig:th_alpha_opsize}. Regarding total throughput, the DCTS scheme with optimal packet sizes achieves an average improvement of $14.2\%$, compared with the case of maximum packet sizes ($1500$ bytes).

\begin{figure}
\centering
\includegraphics[width=0.7\linewidth]{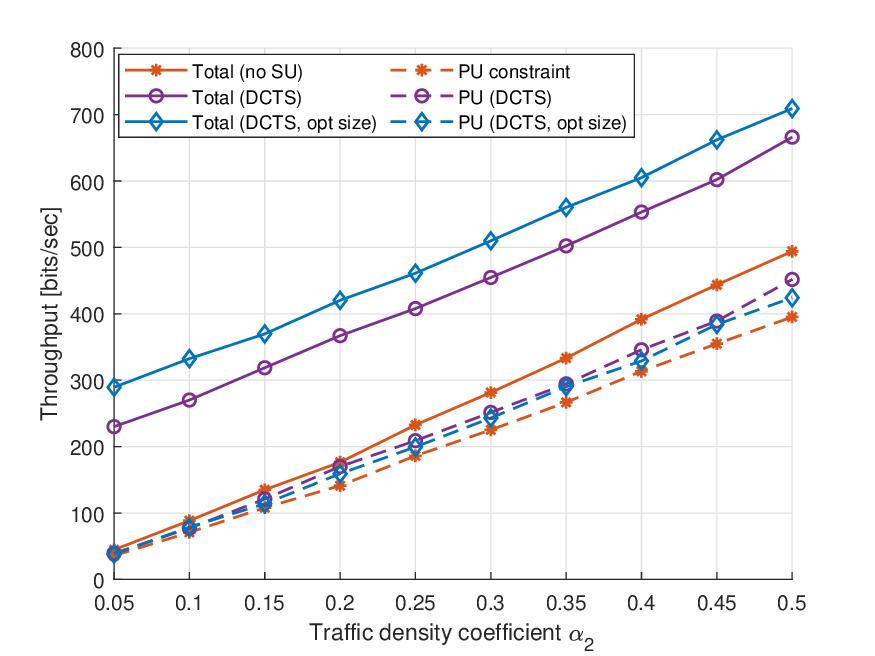}
\caption{Throughput comparison between DCTS scheme with optimized packet size and with normal size, with varying traffic density coefficient $\alpha_2$, $\beta = 0.8$.}
\label{fig:th_alpha_opsize}
\end{figure}

Similarly, we evaluate the performance gain versus PU throughput degradation coefficient in Fig. \ref{fig:th_beta_opsize}. The DCTS scheme with optimal packet sizes provides a $9.9\%$ improvement in total throughput, compared with the case with maximum packet sizes, when the optimal sizes are different from the maximum size. Notice that when degradation coefficient $\beta \leq 0.75$, optimizing SU packet sizes does not provide performance gain. This result is expected since, when large interference is acceptable, the SUs behave more aggressively, by using the largest packet size. On the other hand, when interference is strictly constrained ($\beta \geq 0.8$), using a smaller packet size enables the SUs to avoid interference with part of the PUs. In this way, the SUs trade off packet size for more transmission opportunities, to achieve a higher throughput.

\begin{figure}
\centering
\includegraphics[width=0.7\linewidth]{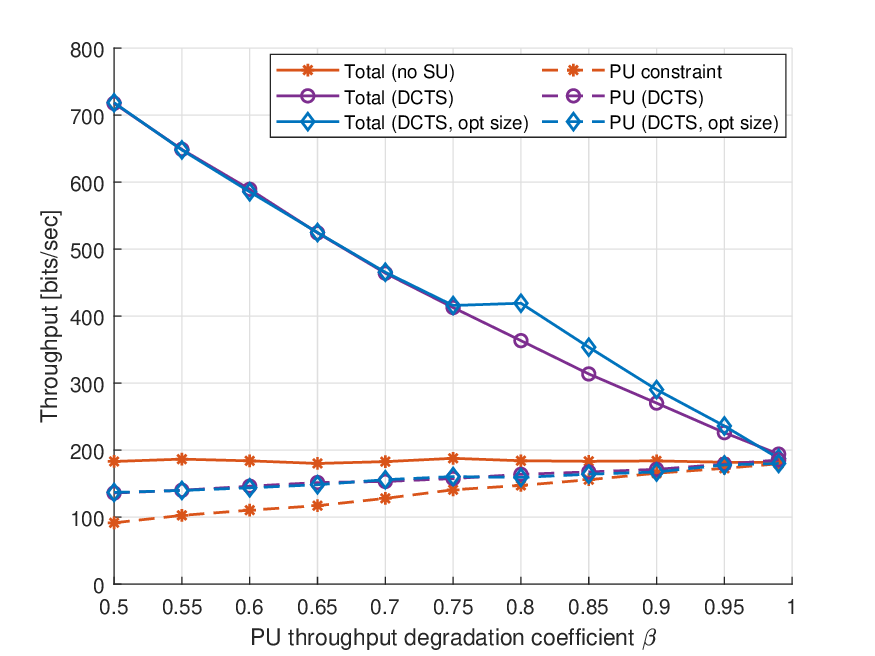}
\caption{Throughput comparison between DCTS scheme with optimized packet size and with normal size, with varying degradation coefficient $\beta$, $\alpha_1 = 0.05$, $\alpha_2 = 0.2$.}
\label{fig:th_beta_opsize}
\end{figure}

\subsection{Interference-free Schedule}
Figure \ref{fig:spec_eff_alpha} plots spectrum efficiency as a function of the traffic density coefficient $\alpha_2$ with $\alpha_1=\alpha_2/4$. The PU throughput degradation coefficient $\beta$ is fixed at $0.8$. We observe that with varying traffic density, the IA scheme is better only when traffic density is low. This result is reasonable as interference-free is easier to achieve for low traffic conditions.

Figure \ref{fig:spec_eff_beta} plots spectrum efficiency as a function of PU throughput degradation coefficient $\beta$ where the traffic density coefficients are fixed as $\alpha_1 = 0.05$ and $\alpha_2 = 0.2$. With varying degradation coefficient, the performance of the IA scheme is better only for large beta values, corresponding to strict PU constraint.

Even though coordination with PUs helps with reducing uncertainty, pursuing interference-free communication limits the performance gain of scheduling. The proposed DCTS scheme can achieve better performance in most system settings. Further, in a practical cognitive system where coordination with PUs is not possible, interference-free approach is not applicable.

\subsection{Time Division versus Frequency Division}
In both Fig. \ref{fig:spec_eff_alpha} and Fig. \ref{fig:spec_eff_beta}, we observe that the proposed DCTS scheme outperforms the DCTS-FDM scheme while the C-TDM scheme is superior to the C-FDM scheme. There are several factors contributing to the worse performance of the FDM type of scheduling schemes. First, the high path attenuation in underwater acoustic channels restricts the interference range. In some regions, constructing sub-channels is not meaningful if only a few users compete for the spectrum. Second, the unique scheduling opportunity provided by the significant propagation delay exists in time domain. Third, because the quality of underwater acoustic channel is highly related to the operating frequency, performance of the network is dominated by the channel with the worst conditions. Lastly, frequency division multiplexing requires guard bands between sub-channels, which lowers the spectral efficiency.

\begin{figure}
\centering
\includegraphics[width=0.7\linewidth]{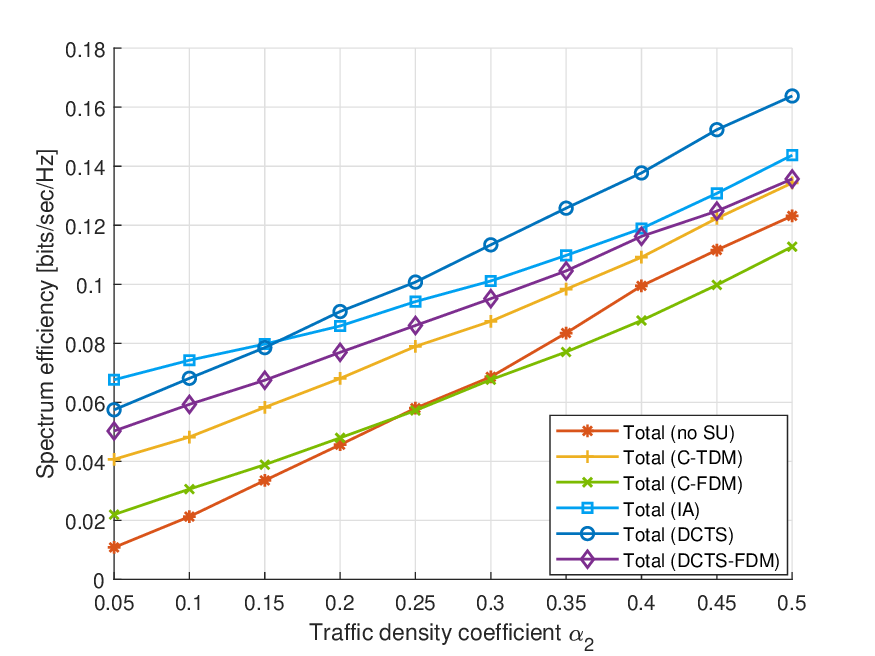}
\caption{Spectrum efficiency comparison between DCTS, DCTS-FDM, C-TDM and C-FDM schemes, with varying traffic density coefficient $\alpha_2$, $\beta = 0.8$.}
\label{fig:spec_eff_alpha}
\end{figure}

\begin{figure}
\centering
\includegraphics[width=0.7\linewidth]{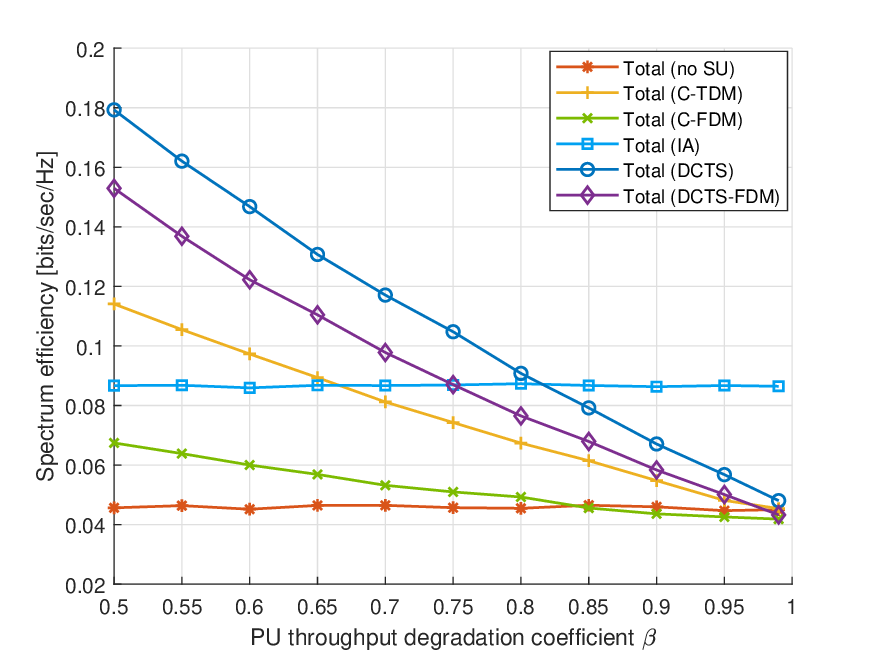}
\caption{Spectrum efficiency comparison between DCTS, DCTS-FDM, C-TDM and C-FDM schemes, with varying degradation coefficient $\beta$, $\alpha_1 = 0.05$, $\alpha_2 = 0.2$.}
\label{fig:spec_eff_beta}
\end{figure}

\subsection{Influence of Other Parameters}
The spectrum efficiency is plotted as a function of the center frequency in Fig. \ref{fig:spec_eff_fc}. Path attenuation in an underwater acoustic channel is highly distance and frequency dependent. Strength of both desired signal and interference are weaker in higher frequency spectrum. As the center frequency changes from $34$ \unit{kHz} to $36$ \unit{kHz}, the DCTS-FDM scheme has better performance, because the change in SINR is dominated by the reduction in interference strength. However, as the center frequency increases, the attenuation of the desired signal dominates, and performance degrades quickly.

\begin{figure}
\centering
\includegraphics[width=0.7\linewidth]{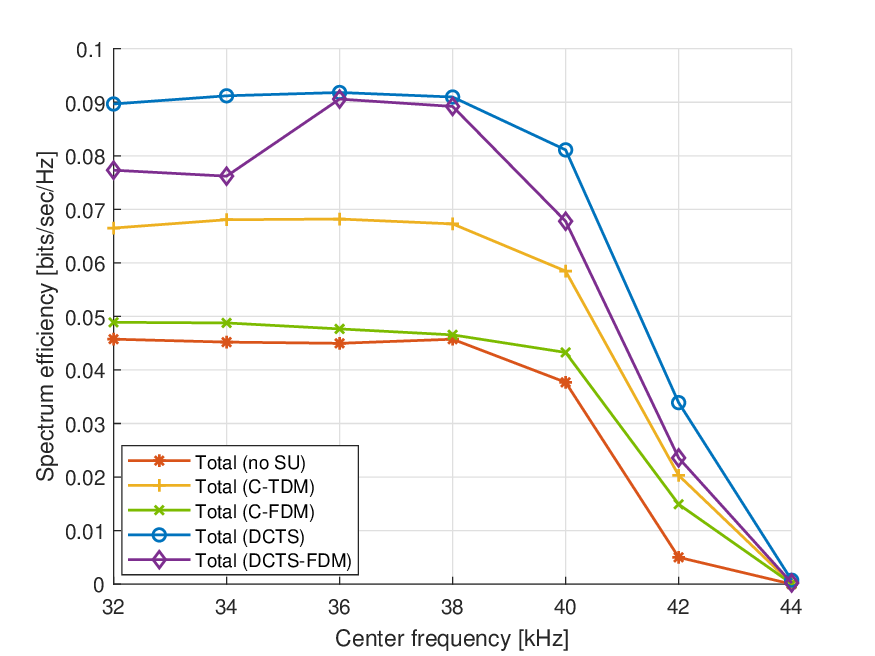}
\caption{Spectrum efficiency comparison between DCTS, DCTS-FDM, C-TDM and C-FDM schemes, with varying center frequency, $\alpha_1 = 0.05$, $\alpha_2 = 0.2$, $\beta = 0.8$.}
\label{fig:spec_eff_fc}
\end{figure}

Figure \ref{fig:spec_eff_dis} plots spectrum efficiency as a function of the end-to-end distance. We observe that for the range ($8$ \unit{km} - $13$ \unit{km}), the proposed schemes works well. However, when the distance becomes larger, reliable communication is not supported and the spectrum efficiency drops to zero rapidly for all schemes.

\begin{figure}
\centering
\includegraphics[width=0.7\linewidth]{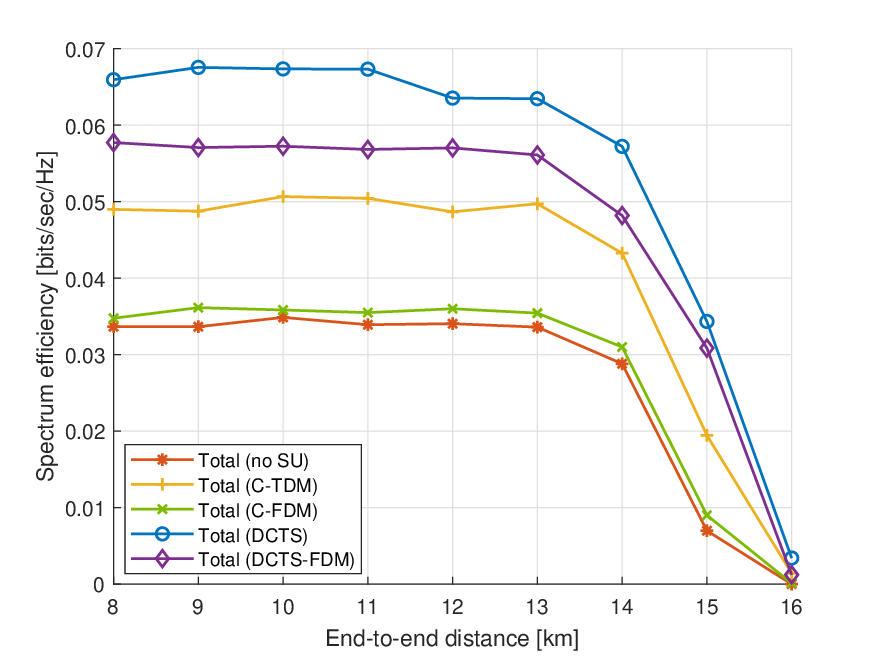}
\caption{Spectrum efficiency comparison between DCTS, DCTS-FDM, C-TDM and C-FDM schemes, with varying end-to-end distance, $\alpha_1 = 0.05$, $\alpha_2 = 0.2$, $\beta = 0.8$.}
\label{fig:spec_eff_dis}
\end{figure}

Lastly, to illustrate that the proposed scheme is widely applicable, we consider a different topology where relays are not in a straight line, with non-uniform distances, as shown in Fig. \ref{fig:topo_nonunif}. The spectrum efficiency is plotted as a function of degradation coefficient in Fig. \ref{fig:spec_eff_beta_topo}, with the traffic coefficients fixed as $\alpha_1 = 0.05$ and $\alpha_2 = 0.2$. We observe that the topology change does not measurably influence the superiority of the proposed scheme over comparison schemes.

\begin{figure}[!t]
\centering
\includegraphics[width=0.4\linewidth]{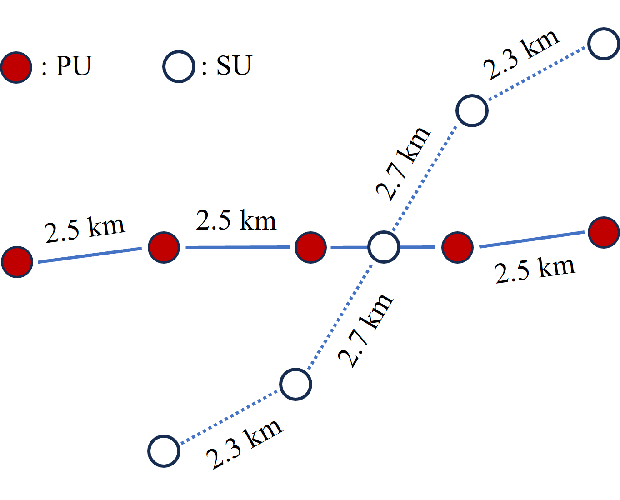}
\caption{A cognitive multi-hop underwater acoustic network with non-uniform relaying distances.}
\label{fig:topo_nonunif}
\end{figure}

\begin{figure}
\centering
\includegraphics[width=0.7\linewidth]{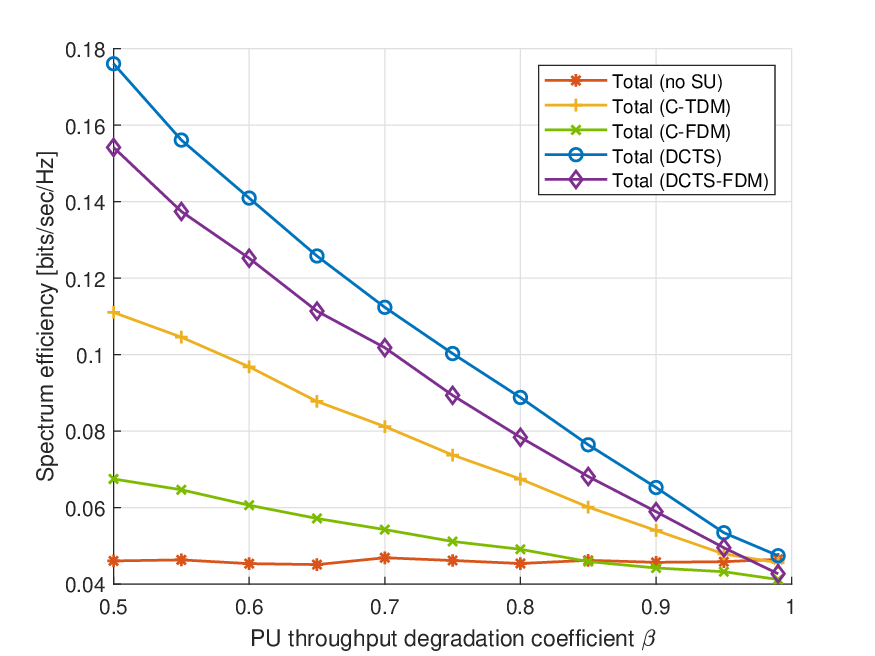}
\caption{Spectrum efficiency as a function of degradation coefficient $\beta$, for a UAN topology with non-uniform relaying distances.}
\label{fig:spec_eff_beta_topo}
\end{figure}

\section{Conclusions} \label{CC}
In this paper, we addressed the scheduling problem of a CM-UAN with an interference constraint, for both centralized and decentralized cases. We first propose a model for system behavior which captures the unique characteristics of UAC via a Markov Chain. The centralized scheduling problem is formulated as a constrained POMDP problem while the decentralized scheduling problem is formulated as a constrained decentralized POMDP problem. The optimal dynamic programming solutions for both centralized and decentralized scheduling are investigated and shown to be computationally expensive. By proving key properties of the objective functions, we propose two approximate schemes, CCTS and DCTS. Both schemes uniquely exploit the strong attenuation in UACs to enable increased sharing of the temporal resource. Furthermore, the impact of packet size on interference is investigated, enabling the optimization of the SU packet sizes in the CM-UAN. The performance of our schemes is evaluated via numerical results, which shows that they are well matched to the underwater channel and provide significant throughput gain with limited loss to the primary users. The results indicate that cognitive time-slot allocation can strongly improve throughput over frequency-slot allocation which is typically considered in radio-frequency based cognitive networks. Under certain traffic conditions, the throughput gain obtained by the proposed DCTS scheme over frequency-slot allocation schemes can be as high as $50\%$.
\bibliographystyle{IEEEtran}
\bibliography{reference.bib}

\end{document}